\documentclass[a4paper,11pt]{article}

\usepackage {amsmath,amssymb,latexsym}
\usepackage {delarray,graphics,graphpap,alltt,graphicx}
\usepackage {a4wide}

\newcommand{\term}{\, \mathrm{term}}

\newcommand{\roC}{\rho _C}

\newtheorem{theorem}{Theorem}[section]

\newtheorem{lemma}[theorem]{Lemma}

\newtheorem{claim}[theorem]{Claim}
\newtheorem{remark}[theorem]{Remark}

\newenvironment{definition}
{ {\noindent {\bf Definition.}} } {  }

\newenvironment{notation}
{ {\noindent {\bf Notation.}} } {  }

\newenvironment{proposition*}[1]
{ {\noindent {\bf Proposition #1.}} } {  }

\newenvironment{proof}
{
{\noindent {\bf Proof.}}}%
{ \hfill $\Box$ }

\newenvironment{proofof}[1]
{
{\noindent {\bf Proof #1.}%
}} { \hfill $\Box$ }

\newcommand{\bool}[1]{\{ 0,1\}^{#1}}

\newcommand{\N}{\mathbb{N}}

\def\squareforqed{\hbox{\rlap{$\sqcap$}$\sqcup$}}
\def\qed{\ifmmode\squareforqed\else{\unskip\nobreak\hfil
\penalty50\hskip1em\null\nobreak\hfil\squareforqed
\parfillskip=0pt\finalhyphendemerits=0\endgraf}\fi}

\newcommand{\pd}{\mathrm{PD}}
\newcommand{\ipd}{\mathrm{invPD}}

\newcommand{\comment}[1]{}
\newcommand{\pl}{\mathrm{plogon}}

\begin{document}


\title{Polylog space compression, pushdown compression, and Lempel-Ziv are incomparable}
\author{Elvira Mayordomo\footnote{Departamento de Inform\'atica e Ingenier\'ia de Sistemas, Instituto de Investigaci\'on en Ingenier\'{\i}a de Arag\'on (I3A), Mar\'ia de Luna 1, Universidad de Zaragoza, 50018 Zaragoza, SPAIN.
elvira(at)unizar.es}\ \footnote{Research supported in part
by Spanish Government MEC and the European Regional
Development Fund (ERDF) under
   Projects TIN2005-08832-C03-02 and TIN2008-06582-C03-02.}\and Philippe Moser\footnote{Department of Computer Science, National University
of Ireland Maynooth, Maynooth, Co. Kildare, Ireland.
pmoser(at)cs.nuim.ie}\ \footnotemark[2]\and Sylvain
Perifel\footnote{LIAFA (Universit\'e Paris Diderot -- Paris
7, CNRS), Paris, France.
sylvain.perifel(at)liafa.jussieu.fr}}

\date{\today}


    \maketitle


\begin{abstract}
The pressing need for efficient compression schemes for XML
documents has recently been focused on stack computation
\cite{HarSha06, LeaEng07}, and in particular calls for a
formulation of information-lossless stack or pushdown
compressors that allows a formal analysis of their
performance and a more ambitious use of the stack in XML
compression, where so far it is mainly connected to parsing
mechanisms. In this paper we introduce the model of
pushdown compressor, based on pushdown transducers that
compute a single injective function while keeping the
widest generality regarding stack computation.

We also consider online compression algorithms that use at
most polylogarithmic space (plogon). These algorithms
correspond to compressors in the data stream model.

We compare the performance of these two families of
compressors with each other and with the general purpose
Lempel-Ziv algorithm. This comparison is made without any a
priori assumption on the data's source and considering the
asymptotic compression ratio for infinite sequences. We
prove that in all cases they are incomparable.
\end{abstract}


\section{Introduction}
The compression algorithms that are required for today
massive data applications necessarily fall under very
limited resource restrictions. In the case of the data
stream setting, the algorithm receives a stream of elements
one-by-one and can only store a brief summary of them, in
fact the amount of available memory is far below linear
\cite{AlMaSz99,IndWoo05}. In the context of XML data bases
the main limiting factor being document size renders the
use of syntax directed compression particularly
appropriate, i.e. compression centered on the grammar-based
generation of XML-texts and performed with stack memory
\cite{HarSha06,LeaEng07}.

In this paper we introduce and formalize useful compression
mechanisms that can be implemented within low
resource-bounds, namely pushdown compressors and
polylogarithmic space online compression algorithms. We
compare these two with each other and with the general
purpose Lempel Ziv algorithm \cite{LemZiv78}.

Finite state compressors were extensively used and studied
before the celebrated result of Lempel and Ziv
\cite{LemZiv78} that their algorithm is asymptotically
better than any finite-state compressor. However, until
recently the natural extension of finite-state to pushdown
compressors has received much less attention, a situation
that has changed due to new specialized compressors for
XML. The work done on stack transducers has been basic and
very connected to parsing mechanisms. Transducers were
initially considered by Ginsburg and Rose in
\cite{GinRos66} for language generation, further corrected
in \cite{GinRos68}, and summarized in \cite{AuBeBo07}. For
these models the role of nondeterminism is specially useful
in the concept of $\lambda$-rule, that is a transition in
which a symbol is popped from the stack without reading any
input symbol.

We introduce here the concept of pushdown compressor as the
most general stack transducer that is compatible with
information-lossless compression. We allow the use of
$\lambda$-rules while having a deterministic (unambiguous)
model. The existence of endmarkers is also allowed, since
it allows the compressor to move away from mere prefix
extension. A more feasible model will also be considered
where the pushdown compressor is required to be invertible
by a pushdown transducer (see Section
\ref{s.pushdown.compressors}). As mentioned before, stack
compression is especially adequate for XML-texts and has
been extensively used \cite{HarSha06,LeaEng07}. We will
also consider  an even more restrictive computation model,
known as visibly pushdown automata
\cite{AluMad06,KuMaVi07}, on which XML compression can be
performed.

Polylogarithmic space online compressors (plogon) are
compression algorithms that use at most polylogarithmic
memory while accessing the input only once. This type of
algorithms models the compression that can actually be
performed in the setting of data streams, where sublinear
space bounds and online input access are assumed, with
constant and polylogarithm being the main bounds
\cite{AlMaSz99,IndWoo05}.

For the comparison of different compression mechanisms we
consider asymptotic  compression ratio for infinite
sequences, and without any a priori assumption on the
data's source. Notice that this excludes results that
assume a certain probability distribution on the data, for
instance the fact that under an ergodic source, the
Lempel-Ziv compression coincides exactly with the entropy
of the source with high probability on finite inputs
\cite{LemZiv78}. This last result is useful when the data
source is known, but it is not informative for arbitrary
inputs, i.e. when the data source is unknown (notice that
an infinite sequence is Lempel-Ziv incompressible with
probability one). Therefore for the comparison of
compression algorithms on general sequences, either an
experimental or a formal approach is needed, such as that
used in \cite{LatStr97}. In this paper we follow
\cite{LatStr97} using a worst case approach, that is, we
consider asymptotic performance on
every infinite sequence.

We prove that the performance of plogon compressors,
pushdown compressors  and Lempel-Ziv's compression scheme
is incomparable in the strongest sense. For each two of
these three mechanisms  we construct a sequence that is
compressed optimally in one scheme but is not in the other,
and vice-versa. In all cases the separation is the
strongest possible, i.e. optimal compressibility is
achieved in the worst case (i.e. almost all prefixes of the
sequence are optimally compressible), whereas
incompressibility is present even in the best case (i.e.
only finitely many prefixes of the sequence are
compressible).

For the comparison of  pushdown transducers with both
plogon and Lempel Ziv, we use the most general pushdown
model (where the pushdown compressor need not be invertible
by a pushdown transducer) for incompressibility and the
more restrictive (where the pushdown compressor is required
to be invertible by a pushdown transducer) for
compressibility, thus obtaining the tightest results.

The proofs are interesting by themselves, since the
witnesses of each of the separations proved show the
strengths and drawbacks of each of the compression
mechanisms. For instance pushdown compressors cannot take
advantage of patterns, while Lempel-Ziv algorithm
compresses well even non correlative repetitions, and
plogon machines require extra information to compress this
kind of data.

This paper contains a revised version of the results  in
\cite{AMMP08} and \cite{MayMos09}.

The paper is organized as follows. Section \ref{sec:prel}\
contains some preliminaries. In section \ref{sec:models},
we present pushdown compressors and  plogon compressor
along with  some basic properties and notations, as well as
a review of the Lempel-Ziv (LZ78) algorithm. In section
\ref{sec:comparison} we present our main results. We end
with a brief conclusion on connections and consequences of
these results for effective dimension and prediction
algorithms.

\section{Preliminaries}\label{sec:prel}

    Let us fix some notation for  strings and languages. Let $\Sigma$ be finite alphabet with at least two
    symbols. W.lo.g. we assume that $0, 1 \in\Sigma$.
    A \emph{string} is an element of $\Sigma^{n}$ for some integer $n$ and a \emph{sequence} is an element of $\Sigma^{\infty}$.
    For a string $x$, its length is denoted by $|x|$.
    If $x,y$ are strings, we write $x\leq y$ (called lexicographic order) if $|x|<|y|$ or $|x|=|y|$ and $x$ precedes $y$ in alphabetical order.
    The empty string is
    denoted by $\lambda $. For $S$ $\in $ $\Sigma ^\infty $ and $i,j$
    $\in $ $\mathbb{N}$, we write $S[i..j]$ for the string consisting of
    the $i^{\textrm{th}}$ through $j^{\textrm{th}}$ symbols of $S$, with
    the convention that $S[i..j]=\lambda $ if $i>j$, and $S[1]$ is the
    leftmost symbol of $S$.
    We say string $y$ is a prefix of string (sequence) $x$, denoted $y\sqsubset x$, if there exists a string (sequence) $a$ such that
    $x=ya$. For a string $x$, $x^{-1}$ denotes $x$ written in reverse order.
    For a function $f:A\rightarrow B$, $f(x)=\perp$ means $f$ is not defined on input $x$.
    For a sum $\sum_{j=1}^n a_j$ let $\term(k)$ denote $a_k$. For a function $f$, $f^{(2)}$ denotes $f\circ f$.



    Given a sequence $S$ and a function $T:\Sigma ^* \to \Sigma ^*$,
    the $T$- upper and lower compression ratios of $S$  are given by
    \begin{align*}
        \rho_{T}(S) &= \liminf_{n\rightarrow\infty} \frac{|T(S[1\ldots n])|}{n}     , \text{ and}\\
        R_{T}(S) &= \limsup_{n\rightarrow\infty} \frac{|T(S[1\ldots n])|}{n}.
    \end{align*}

\begin{notation}
We use $K(w)$ to denote the standard (plain) Kolmogorov
complexity, that is, fix a universal Turing Machine $U$.
Then for each string $w\in\Sigma^*$,
\[K(w) = \min\{|p|\,|\,p\in\bool{*}, U(p)=w\}\] i.e.,
$K(w)$ is the size of the shortest binary program that
makes $U$ output $w$. Although some authors use $C(w)$ to
denote (plain) Kolmogorov complexity, we reserve this
notation to denote  a particular compression algorithm $C$
on input $w$.
\end{notation}

\section{Compressors with low resource-bounds}\label{sec:models}

In this section we consider several families of lossless
compression methods that use very low computing resources.
We introduce a detailed definition of stack-computable
compressors together with some variants and review
poly-logarithmic space computable compressors and the
celebrated Lempel-Ziv algorithm.

\subsection{Pushdown compressors}\label{s.pushdown.compressors}
We discuss next different formalizations of information
lossless compressors that are equipped with stack memory.
The most general ones are allowed to use a bounded number
of lambda-rules, that is, stack movements that don't
consume an input symbol. The most restricted pushdown
compressors we consider here are
 visibly pushdown automata that are suitable for XML
compression.

There are several natural variants for the model of
pushdown transducer \cite{AuBeBo07}, both allowing
different degrees of nondeterminism and computing partial
(multi)functions by requiring final state or empty stack
termination conditions. But our purpose here is to compute
a total and well-defined (single valued) function,
therefore nondeterminism should be very limited and natural
termination conditions are equivalent.

The main variants that will influence the computing power
of a pushdown compressor while remaining information
lossless are the presence of lambda-rules, the possible
restrictions of stack movements, and the use of an
endmarker, that is an extra symbol signaling the end of the
finite input.

We will introduce here pushdown compressors, invertible
pushdown compressors, and visibly pushdown compressors
(this last one defined in \cite{AluMad06,KuMaVi07}).

The definitions below are adapted from those in
\cite{AMMP08,MayMos09}.

\begin{definition}A {\itshape bounded pushdown compressor} (BPDC)
is an 8-tuple
\begin{center}
    $C=(Q, \Sigma , \Gamma , \delta , \nu , q_0,  z_0, c)$
\end{center}
where\begin{itemize}
\item $Q$ is a finite set of states
\item $\Sigma$ is the finite input/output alphabet
       \item $\Gamma $ is the finite stack alphabet
       \item $\delta :Q\times (\Sigma \cup \{\lambda \})\times \Gamma \rightarrow Q\times \Gamma
       ^*$ is the transition function
       \item $\nu :Q\times \Sigma \times \Gamma \rightarrow
       \Sigma ^*$ is the output function
       \item $q_0$ $\in $ $Q$ is the initial state
       \item $ z_0$ $\in $ $\Gamma $ is the start stack symbol
       \item $c\in\N$ is an upper bound on the number of
       $\lambda$-rules per input symbol.
     \end{itemize}
\end{definition}

We use $\delta _Q$ and  $\delta _{\Gamma ^*}$ for the projections of  
function $\delta$. We restrict $\delta $ so that $ z_0$
cannot be removed from the stack bottom, that is, for every
$q$ $\in $ $Q$, $b$ $\in $ $\Sigma \cup \{\lambda \}$,
either $\delta (q,b, z_0)=\perp $, or $\delta (q,b,
z_0)=(q',v z_0)$, where $q'$ $\in $ $Q$ and $v$ $\in $
$\Gamma ^*$.

Note that the transition function $\delta $ accepts
$\lambda $ as an input character in addition to elements of
$\Sigma$, which means that $C$ has the option of not
reading an input character while altering the stack, such a
movement is called a {\sl $\lambda$-rule}. In this case
$\delta (q, \lambda, a)=(q', \lambda )$, that is, we pop
the top symbol of the stack. To enforce determinism, we
require that at least one of the following hold for all $q$
$\in $ $Q$ and $a$ $\in $ $\Gamma $:\begin{itemize}
            \item $\delta (q,\lambda ,a)=\perp$,
            \item $\delta (q,b,a)=\perp $ for all $b$ $\in $ $\Sigma$.
            \end{itemize}
            We {\sl restrict the number of $\lambda$-rules} that can be
            applied as follows: between the input symbols in positions $n$
            and $n+1$ a maximum of
            $c$ $\lambda$-rules can be applied.

We first consider the transition function $\delta$ as
having inputs in  $Q\times (\Sigma \cup \{\lambda \})\times
\Gamma^+$, meaning that only the top symbol of the stack is
relevant. Then we use the extended transition function
$\delta ^{*}:Q\times \Sigma^* \times \Gamma ^+\rightarrow
Q\times \Gamma^*$, defined recursively as follows. For $q$
$\in $ $Q$, $v$ $\in $ $\Gamma ^+$, $w$ $\in $ $\Sigma^*$,
and $b$ $\in $ $\Sigma$
\begin{center}
    $\delta ^{*}(q, \lambda , v)=$$\left\{
                                      \begin{array}{ll}
                                        \delta ^{*}(\delta_Q(q, \lambda , v),
\lambda , \delta _{\Gamma ^*}(q, \lambda , v)), & \hbox{if $\delta (q, \lambda , v)\neq \perp $;} \\
                                        (q, v), & \hbox{otherwise.}
                                      \end{array}
                                    \right.$
\end{center}
\begin{center}  
     $\delta ^{*}(q, wb, v)=$$\left\{
                               \begin{array}{ll}

\delta^*(\delta_Q(\delta^*_Q(q, w, v), b, \delta^*_{\Gamma
^*}(q, w, v)), \lambda, \delta_{\Gamma ^*}(\delta^*_Q(q, w,
v), b, \delta^*_{\Gamma ^*}(q, w, v)) ),  \\\ \ \ \ \ \ \
\hbox{if $\delta ^*(q, w, v)\neq \perp $ and
$\delta(\delta^*_Q(q, w, v), b, \delta^*_{\Gamma ^*}(q,
w, v))\neq \perp $;} \\
                                 \perp , \ \ \  \hbox{otherwise.}
                               \end{array}
                             \right.$
\end{center}

That is, $\lambda $-rules are implicit in the definition of
$\delta ^{*}$. We abbreviate $\delta ^{*}$ to $\delta $,
and $\delta (q_0,w, z_0)$ to $\delta (w)$. We define the
{\itshape output} from state $q$ on input $w \in \Sigma ^*$
with $z \in \Gamma ^*$ on the top of the stack by the
recursion $\nu (q,\lambda ,z)= \lambda ,$
\begin{center}
    $\nu(q,wb,z) = \nu(q,w,z)\, \nu(\delta _Q(q,w,z), b, \delta _{\Gamma ^*}(q,w,z)).$
\end{center}
The {\itshape output} of the compressor $C$ on input $w$
$\in $ $\Sigma ^*$ is the string $C(w)=\nu (q_0,w, z_0)$.

The input of an information-lossless compressor can be
reconstructed from the output and the final state reached
on that input.

\begin{definition} A BPDC $C=(Q, \Sigma , \Gamma , \delta , \nu
, q_0,  z_0, c)$ is {\itshape information-lossless} (IL) if
the function
\begin{eqnarray*}
    \Sigma ^*&\rightarrow &\Sigma ^*\times Q\\
w&\mapsto &(C(w),\delta _Q(w))\\
\end{eqnarray*}
is one-to-one. An {\itshape information-lossless pushdown
compressor} (ILPDC) is a BPDC that is IL.
\end{definition}

Intuitively, a BPDC {\itshape compresses} a string $w$ if
$|C(w)|$ is significantly less than $|w|$. Of course, if
$C$ is IL, then not all strings can be compressed. Our
interest here is in the degree (if any) to which the
prefixes of a given sequence $S$ $\in $ $\Sigma ^\infty $
can be compressed by an ILPDC.

We will also consider PDC that have endmarkers, a
characteristic that can achieve a better compression rate.

\begin{definition} An {\itshape information-lossless pushdown
compressor with endmarkers} (ILPDCwE) is a BPDC $C=(Q,
\Sigma\cup\{\$\} , \Gamma , \delta , \nu , q_0,  z_0, c)$
with input alphabet $\Sigma\cup\{\$\}$ ($\$\not\in\Sigma$)
such that the function
\begin{eqnarray*}
    \Sigma ^*&\rightarrow &\Sigma ^*\times Q\\
w&\mapsto &(C(w\$),\delta _Q(w))\\
\end{eqnarray*}is one-to-one.
\end{definition}

Notice that the use of endmarkers can improve compression.
In particular each ILPDC is a particular case of  ILPDC
with endmarkers, but there are ILPDC with endmarkers that
perform better than usual ILPDC.

We will denote as pushdown compression ratio the concept
corresponding to the most general family of pushdown
compressors, those that use endmarkers.

\begin{notation} The best-case {\itshape pushdown compression ratio}
of a sequence $S$ $\in $ $\Sigma ^\infty $ is
 $\rho_{PD}(S)=$ $\inf \{\roC(S)\mid \mbox{ $C$ is an ILPDCwE}\}$.

{\noindent The worst-case {\itshape pushdown compression
ratio} of a sequence $S$ $\in $ $\Sigma ^\infty $ is
    $R_{PD}(S)=$ $\inf \{R_C(S)\mid \mbox{ $C$ is an ILPDCwE}\}$.}
\end{notation}

Notice that so far we have not required that the
computation should be invertible by another pushdown
transducer, which is a natural requirement for practical
compression schemes. The standard PD compression model does
not guarantee the decompression to be feasible and it is
currently
    not known whether the exponential time brute force inversion can even be improved to polynomial time.
To guarantee both decompression and compression to be
feasible, we require the existence of a PD machine
    that given the compressed string (and the final state), outputs the decompressed one. This yields two PD compression schemes, the standard one (PD)
    and invertible PD. Contrary to Finite State computation, it is not known whether both are equivalent. This is by no means a limitation, since all results
    in this paper are always stated in the strongest form, i.e. we obtain results of the form ``X beats PD'' and ``invertible PD beats X''.

    Here is the definition of invertible PD compressors. We want
    this definition to be the most restrictive one and therefore
    regular ILPDC.

    \begin{definition}
        $(C,D)$ is an invertible PD compressor (denoted $\ipd$) if $C$ is an ILPDC and $D$ is a PD transducer s.t.
        $D(C(w), \delta_Q(w))=w$, i.e. $D$, given both $C(w)$ and the final state, outputs $w$.
    \end{definition}

\begin{notation} The best-case {\itshape invertible pushdown compression ratio}
of a sequence $S$ $\in $ $\Sigma ^\infty $ is
 $\rho_{\ipd}(S)=$ $\inf \{\roC(S)\mid \mbox{ $C$ is an }\ipd\}$.

{\noindent The worst-case {\itshape invertible pushdown
compression ratio} of a sequence $S$ $\in $ $\Sigma ^\infty
$ is
    $R_{\ipd}(S)=$ $\inf \{R_C(S)\mid \mbox{ $C$ is an } \ipd\}$.}
\end{notation}

\comment{ In section \ref{section:last} in order to make
our result stronger we will also consider PDC that have
endmarkers, a characteristic that can achieve a better
compression rate. In this case the information lossless
(IL) condition is relaxed to the function
\begin{eqnarray*}
    \Sigma ^*&\rightarrow &\Sigma ^*\times Q\\
w&\mapsto &(C(w\$),\delta _Q(w))\\
\end{eqnarray*}
being one-one, for \$ a fixed symbol out of $\Sigma$. }

We end this section with the concept of visibly pushdown
automata from \cite{AluMad06,KuMaVi07} that is extensively
used in the compression of XML.

A {\itshape visibly pushdown compressor} (visiblyPD) is an
information-lossless pushdown compressor  for which  the
input alphabet has three types of symbols, call symbols,
return symbols, and internal symbols. The main restriction
is that while reading a call, the automaton must push one
symbol, while reading a return symbol, it must pop one
symbol (if the stack is non-empty), and while reading an
internal symbol, it can only update its control state.

Therefore the compression ratio attained by visibly
pushdown automata is an upper bound on the compression
ratio attained through the pushdown compressors defined
above.

\subsection{plogon compressors}

 We introduce  the family of compressors that can be computed online with at most poly-logarithmic space.
 Notice that these resource bounds correspond to those of the data stream model \cite{AlMaSz99,IndWoo05},
 where the input size is massive in comparison with the
 available memory, and the input can only be read once.

    \begin{definition}(Hartmanis, Immerman, Mahaney \cite{HaImMa78})
    A Turing machine  $M$ is a plogon transducer if it has the following properties, for each input string $w$
    \begin{itemize}
        \item   the computation of $M(w)$ reads its input from left to right (no turning back),
        \item $M(w)$ is given $|w|$ written in binary (on a special tape),
        \item $M(w)$ writes the output from left to right on a write-only output tape,
        \item $M(w)$ uses memory bounded by $\log(|w|)^c$, for a constant $c$.
    \end{itemize}
    \end{definition}

We denote with plogon the class of plogon transducers.

    Note that contrary to Finite State transducers, a plogon transducer is not necessarily a mere extender, i.e., there is a plogon transducer $M$ and
    strings $w, x$ such that $M(wx)\not\sqsupset M(w)$.

    \begin{definition}
    A plogon transducer $C:\Sigma^{*}\rightarrow\Sigma^{*}$ is an information lossless compressor (ILplog)
    if it is 1-1.
    \end{definition}

\begin{notation} The best-case {\itshape plogon compression ratio}
of a sequence $S$ $\in $ $\Sigma ^\infty $ is
 $\rho_{\pl}(S)=$ $\inf \{\roC(S)\mid \mbox{ $C$ is an ILplog}\}$.

{\noindent The worst-case {\itshape plogon compression
ratio} of a sequence $S$ $\in $ $\Sigma ^\infty $ is
    $R_{\pl}(S)=$ $\inf \{R_C(S)\mid \mbox{ $C$ is an ILplog}\}$.}
\end{notation}

\subsection{Lempel Ziv compression scheme}

 Let us give a brief description of the classical LZ78 algorithm  \cite{LemZiv78}. Given an input
$x\in\Sigma^*$, LZ parses $x$ in different phrases $x_i$,
    i.e., $x=x_1 x_2 \ldots x_n$ ($x_i\in\Sigma^*$) such that
    every prefix $y\sqsubset x_i$, appears before $x_i$ in the parsing (i.e. there exists $j<i$ s.t. $x_j=y$).
    Therefore for every $i$, $x_i = x_{l(i)}b_i$ for $l(i)<i$ and
    $b_i\in\Sigma$. We sometimes denote the number of phrases in the
    parsing of $x$ as $P(x)$.
After step $i$ of the algorithm, the $i$ first phrases
$x_1,\dots,x_i$ have been parsed and stored in the
so-called \emph{dictionary}. Thus, each step adds one word
to the dictionary.

    LZ encodes $x_i$ by a prefix free encoding of ${l(i)}$ and the      
    symbol $b_i$, that is, if $x=x_1 x_2 \ldots x_n$ as before, the
    output of LZ on input $x$ is
    \[LZ(x)= c_{l(1)}b_1 c_{l(2)}b_2 \ldots c_{l(n)}b_n\] where
    $c_i$ is a prefix-free coding of $i$ (and $x_0=\lambda$).

    For a string $z=xy$ we denote by $LZ(y|x)$ the output of LZ on $y$ after having read $x$ already.

    LZ is usually restricted to the binary alphabet, but the
    description above is valid for any alphabet $\Sigma$.

   \comment{ Let us give a brief description of the classical LZ78 algorithm \cite{LemZiv78}. Fix a set of code words $C$, i.e.
    a set of strings $C=\{ c_1 < c_2 < \ldots \}$ such that no $c_i$ is a prefix of a string
    $c_j$ (for any $i,j\in\N$). Given an input $x\in\bool{*}$, LZ parses $x$ in phrases $x_i$'s,
    i.e. $x=x_1 x_2 \ldots x_n$ ($x_i\in\bool{*}$) where if $x_i$ is a phrase in the parsing,
    then every prefix $y\sqsubset x_i$, appears before $x_i$ in the parsing. We sometimes denote the
    parsing separating the phrases by commas: $x_1, x_2, \ldots ,x_n$.
    Note that the parsing of $x$ is unique. LZ encodes $x$ by replacing every phrase $x_i$ in the parsing
    by a code word $c_{l(i)}$ followed by a bit $b_i$, where $l(i)$ is the position of the longest prefix
    of $x_i$ in the parsing, and $b_i$ is the last bit of the phrase, i.e.
    $x_i = x_{l(i)}b_i$. These completely specify the phrase being encoded.
    We denote the output of LZ on input $x$ by $LZ(x)$.}

\section{The performances of the LZ78 algorithm, plogon compressors and pushdown compressors are incomparable}\label{sec:comparison}

In this section we prove that the two families of
compressors we have introduced,  pushdown and plogon
compressors, and the Lempel Ziv compression scheme, are all
incomparable. That is, for any pair among those three,
there are different individual sequences on which one is
outperformed by the other and vice versa. In all cases we
get low worst-case rate ($\rho$) for one method versus high
best-case rate ($R$) for the other, i.e. the widest
possible separation between them.

\subsection{Lempel Ziv  beats Pushdown compression}

Our first result shows that there is a sequence that our
most general family of pushdown compressors cannot compress
and that is optimally compressible by Lempel Ziv.

The proof is based on two intuitions, that require a
careful analysis. The first one is that from a few
Kolmogorov-random strings a much longer
pushdown-incompressible string can be constructed. On the
other hand, a sequence with enough (and non-consecutive)
repeated substrings can be compressed optimally by
Lempel-Ziv.

\begin{theorem}\label{theo:first}
There exists a sequence $S$ such that $$R_{LZ}(S)=0$$ and
$$\rho _{PD}(S)=1.$$
\end{theorem}

\begin{proof}
    Consider the sequence $S=S_1S_2\ldots$ where $S_n$ is constructed as follows.
    Let $x=x_1 x_2 \ldots x_{n^2}$ ($|x_i|=n$) be a Kolmogorov-random string with $K(x)\geq n^3 \log |\Sigma|$.
    Let
    $$S_n= x_{i_1}\ldots x_{i_{l}}$$
    where $i_j \in\{1,\ldots, n^2\}$ for every $1\leq j \leq l$ are indexes, defined later on.
    Let $$l=\frac{1}{n}\sum_{k=1}^{n}k\min(|\Sigma|^k,n^{\frac{2k}{n}+1})$$
    so that
    \begin{equation} \label{e.sn1}
    |S_n|= nl = \sum_{k=1}^{n}k\min(|\Sigma|^k,n^{\frac{2k}{n}+1}).
    \end{equation}
    Let us show that for every $\epsilon>0$ and for $n$ large enough
    \begin{equation} \label{e.sn2}
    n^{5-\epsilon} \leq |S_n| \leq n^{5}.
    \end{equation}
    We prove the first inequality.
    $$|S_n|= \sum_{k=1}^{n}k\min(|\Sigma|^k,n^{\frac{2k}{n}+1}) \leq n \term(n) \leq n \cdot n \cdot n^{\frac{2n}{n}+1} = n^5. $$
    For the second inequality we have
    \begin{align*}
    |S_n| &= \sum_{k=1}^{n}k\min(|\Sigma|^k,n^{\frac{2k}{n}+1}) \\
    &\geq \sum_{k=(1-\frac{\epsilon}{4}) n}^{n}k\min(|\Sigma|^k,n^{\frac{2k}{n}+1}) \\
    &\geq \frac{n\epsilon}{4} \term((1-\frac{\epsilon}{4}) n) \\
    &\geq n^{5-\epsilon}.
    \end{align*}

        Let $C_1, C_2,\ldots$ be an enumeration of all  ILPDCwE such that
    $C_i$ can be encoded in at most $i$ bits and such that
            a maximum of
            $\log^{(2)}i$ $\lambda$-rules can be applied per symbol.
    The following claim shows that there are many $C$-incompressible
    strings $x_i$.

    \begin{claim}\label{c.many.incompressible}
    Let $F_n=\{C_1,\ldots,C_{\log n}\}$.
    Let $w\in\Sigma^*$.
    \begin{enumerate}
\item Let
    $C\in F_n$.    There are at least $(1-\frac{1}{2\log n})n^2$ strings $x_i$
($1\leq i \leq n^2$) such that
    $$|C(wx_i)|-|C(w)| > n-2\sqrt{n} .$$
    \item There is a  string $x_i$ such that  for every  $C\in F_n$, $$|C(wx_i)|-|C(w)| > n-2\sqrt{n} .$$
\end{enumerate}
    \end{claim}
\begin{proofof}{of Claim \ref{c.many.incompressible}}
After having read $w$,
    $C$ is in state $q$, with stack content $yz$, where $y$ denotes the $n\log^{(2)} n$
    topmost symbols of the stack (if the stack is shorter then $y$ is the whole stack).
    It is clear that while reading an $x_{i}$, $C$ will not pop the stack below $y$.

    Let $T=(1-\frac{1}{2\log n})n^2$, and let $C(q,yz,x_i\$)$
    denote the output of $C$ when started in state $q$ on input $x_i\$$ with stack content $yz$.
    Suppose the claim false, i.e.
    there exist more than $n^2-T$ words $x_i$ such that $C(q,yz,x_i\$)=p_i$, ends in state $q_i$,
    and $|p_i|\leq n-2\sqrt n+O(1)$ (notice that the output on symbol $\$$ is $O(1)$). Denote by $G$ the set of such strings $x_i$.
    This yields the following short  program  for $x$ (coded with alphabet $\Sigma$):
    $$p=(n,C,q,y,a_1t_1a_2t_2\ldots a_{n^2}t_{n^2})$$
    where each comma costs less than $3\log |s|$,   where $s$ is the element between two commas;
    $a_i=1$ implies $t_i=x_i$, $a_i=0$ implies $x_i\in G$ and
    $t_i = d(q_i)01d(|p_i|)01p_i$ (where $d(z)$ for any string $z$, is the string written with every symbol doubled),
    i.e. $|t_i|\leq n - \sqrt n$.
    $p$ is a program for $x$: once $n$ is known, each $a_i t_i$ yields either $x_i$ (if $a_i=1$)
    or $(p_i,q_i)$ (if $a_i=0$). From $(p_i,q_i)$, simulating $C(q,yz,u\$)$ for each $u\in\Sigma^n$
    yields the unique $u=x_i$ such that $C(q,yz,u\$)=p_i$ and ends in state $q_i$. The simulations are possible,
    because $C$ does not read its stack further than $y$, which is given.
    We have
    \begin{align*}
    |p|&\leq O(\log n) + n\log^{(2)}n+(n+1)T+(n^2-T)(n-\sqrt n)\\
    &\leq O(n^2) + n^3 - \frac{n^{2.5}}{2\log n}\\
    &\leq n^3 -\frac{n^{2.5}}{4\log n} \\
    \end{align*}
    which contradicts the randomness of $x$, thus proving part 1. 

    Let $W_j$ be the set of strings $x_i$ that are compressible by $C_j$; by
        1.,
    $|W_j|\leq n^2/2\log n$. Let $R=\{x_{i}\}_{i=1}^{n^2} - \cup_{j=1}^{\log n}W_j$
    be the set of strings incompressible by all $C\in F_n$. We have
    $$|R|\geq n^2 - \log n \cdot n^2 / 2\log n =n^2/2> 1.$$
    This proves part 2.
\end{proofof}

    We finish the definition of $S_n$ by picking
     $x_{i_1}$ to be the first string fulfilling the second part of Claim \ref{c.many.incompressible} for $w=S_1S_2\ldots S_{n-1}$.
    The construction is similar for all strings $\{x_{i_j}\}_{j=2}^{l}$, by taking $w=S_1S_2\ldots S_{n-1}x_{i_1}\ldots x_{i_{j-1}}$, thus ending the construction
    of $S_n$.

    Let us show that $\rho_{PD}(S)=1$. Let $\epsilon>0$.
    Let $C=C_k$ be an ILPDCwE; then for almost every $n$, and for all $0\leq t \leq
    |S_n|/n$, $0\le i<n$
    we have
    \begin{align*}
    &\frac{|C(S_1\ldots S_{n-1}S_n[1\ldots tn+i]\$)|}{|S_1\ldots S_{n-1}S_n[1\ldots tn+i]|}\\
    &\geq   \frac{\sum_{j=k}^{n-1} (j-2\sqrt j)|S_j|/j + t(n-2\sqrt n)-O(1)}{\sum_{j=1}^{n-1}|S_j|+(t+1)n}\\
    &\geq 1 - \frac{\sum_{j=1}^{k-1}|S_j|}{\sum_{j=1}^{n-1}|S_j|+(t+1)n} - 2\frac{\sum_{j=k}^{n-1}|S_j|/\sqrt j}{\sum_{j=1}^{n-1}|S_j|+(t+1)n}
    -2 \frac{t\sqrt n+ n/2}{\sum_{j=1}^{n-1}|S_j|+(t+1)n}\\
    &\geq 1 - \epsilon/4 - O(1) \frac{\sum_{j=1}^{n-1}j^{4.5}}{\sum_{j=1}^{n-1}j^{5-\delta}}  -\epsilon/4 \quad\quad\quad \text{(by Equation \ref{e.sn2})}\\
    &\geq 1 - \epsilon/2 - \frac{O(1)(n-1)\term(n-1)}{\frac{n}{3} \term(\frac{n}{3})}\\
    &\geq 1 - \epsilon/2 - \frac{O(1)(n-1)(n-1)^{4.5}}{\frac{n}{3} (\frac{n}{3})^{5-\delta}}\\
    &\geq 1 - \epsilon/2 - \epsilon/2 = 1 - \epsilon \quad\quad\quad \text{(choosing $\delta=0.1$)}
    \end{align*}
    i.e. $\rho_{PD}(S)=1$.

    We show that $R_{LZ}(S)=0$. Suppose LZ has already parsed input $S_1\ldots S_{n-1}$, and has
    $d_n$ words in its dictionary ($d_n\leq n |S_n|$).
    Let $P$ be the parsing of $S_n$ by LZ,
    let $t_P$ be the size of the largest string in $P$ and let $1\leq k \leq t_P$.
    Let us compute the maximum number of strings of size $k$ in $P$. Any string $u$ of size
    $k$ in a parsing of $S_n$ is of the form
    $$u= x_{t_1}[t\ldots n] x_{t_2}\ldots x_{t_{k/n}}$$
    i.e. amounts to choose $k/n$ strings $x_{t_i}$ and the position $1\leq t \leq n$ where $u$ starts in $x_{t_1}$.
    Therefore there are at most $\#k=n \cdot (n^2)^{k/n}= n^{1+2k/n}$ such words $u$ of size $k$.

    Let $P_w$ be the worst-case parsing of $S_n$, that starts on an empty dictionary and parses all possible strings of size $k$
    in $S_n$ (for every $k\leq t_w$), where $t_w$ is the size of the largest string in $P_w$ i.e.,
    $\min(|\Sigma|^1, n^{1+2/n})$ strings of size one are parsed, followed
    by $\min(|\Sigma|^2,n^{1+4/n})$ strings of size 2, \ldots, followed by $\min(|\Sigma|^k, n^{1+2k/n})$ strings of size $k$, and so on.
    Because
    $$\sum_{k=1}^{n}k\min(|\Sigma|^k,n^{\frac{2k}{n}+1}) = |S_n|$$
    we have $t_w \leq n$.

    Let $p$ (resp. $p_w$) be the number of phrases in $P$ (resp. $P_w$). We have
    $p\leq p_w$, and
    $|LZ(S_n|S_1\ldots S_{n-1})| \leq p\log(p+d_n)$.
    Since
    $$p_w = \sum_{k=1}^{t_w}\min(|\Sigma|^k,n^{\frac{2k}{n}+1} )\leq n\term(n)=n^4$$
    we have
    $$
        |LZ(S_n|S_1\ldots S_{n-1})| \leq n^4 \log (n^4+n|S_n|) \leq n^{4+\alpha}
    $$
    where $\alpha>0$ can be arbitrary small.

    Let $0\leq t \leq |S_n|/n$, $0\le i<n$.
    We have
    \begin{align*}
        \frac{|LZ(S_1\ldots S_{n-1}S_n[1\ldots tn+i])|}{|S_1\ldots S_{n-1}S_n[1\ldots tn+i]|}
        &\leq
        \frac{\sum_{j=1}^{n-1}|LZ(S_j|S_1\ldots S_{j-1})|+|LZ(S_n|S_1\ldots S_{n-1})|}
        {\sum_{j=1}^{n-1}|S_j|}\\
        &\leq\frac{\sum_{j=1}^{n-1}|LZ(S_j|S_1\ldots S_{j-1})|}
        {\sum_{j=1}^{n-1}|S_j|} + \frac{n^{4+\alpha}}
        {\sum_{j=1}^{n-1}|S_j|}\\
        &\leq \frac{\sum_{j=1}^{n-1}j^{4+\alpha}}{\sum_{j=1}^{n-1}j^{5-\delta}} +
        \frac{n^{4+\alpha}}{\sum_{j=1}^{n-1}j^{5-\delta}} \\
        &\leq \epsilon/2 + \epsilon/2 \leq \epsilon
    \end{align*}
    i.e. $R_{LZ}(S)=0$.
\end{proof}

\subsection{Lempel Ziv beats plogon compressors}\label{sec:lz}

The Lempel Ziv algorithm can also surpass plogon
compressors. Our second comparison detects sequences on
which Lempel-Ziv achieves optimal compression whereas a
plogon compressor has the worst possible performance. The
construction is based on repetition of Kolmogorov random
strings. We show that Lempel-Ziv works well on any repeated
pattern, whereas in polylogarithmic space big patterns
cannot be stored.

\begin{theorem}\label{t.lz.1}
    There exists a sequence $S$ such that
    $$R_{LZ}(S) = 0\quad \text{and }\rho_{\pl}(S)=1.$$
\end{theorem}

    The proof will use  the following  general property that bounds   the output of Lempel-Ziv on strings of the form $w=u^n$.

\begin{lemma} \label{LZ is OK_finite}
Let $n \in \mathbb{N}$ and let $ u \in \Sigma ^*$, where $u
\neq \lambda$. Define $l=1+|u|$ and $w=u^n$.
Consider the execution of Lempel-Ziv on $w$ starting from a
dictionary containing $d\geq 0$ phrases.
%
Then we have that
    \begin{equation}\label{e.lz}
    |LZ(w)|\leq \sqrt{2l|w|} \log(d+\sqrt{2l|w|})
    \end{equation}
\end{lemma}
\begin{proofof}{of Lemma \ref{LZ is OK_finite}}
 Let us fix $n$ and
consider the execution of Lempel-Ziv algorithm on $w$: as
it parses the word, it enlarges its dictionary of phrases.
Fix an integer $k$ and let us bound the number of new words
of size $k$ in the dictionary. As the algorithm parses
$|u|$, the number of different words of size $k$ in $u^n$
is at most $|u|$ (at most one beginning at each symbol of
$u$). Therefore we obtain a total of at most $|u|$
different new words of size $k$ in $w$. This total is
bounded from above by $l=|u|+1$.

Therefore at the end of the algorithm and for all $k$, the
dictionary
  contains at most $l$ new words of size $k$. We can now  bound from above the size of
  the compressed image of $w$. Let $p$ be the number of new phrases in the
  parsing made by Lempel-Ziv algorithm. The size of the compression is
  then $p\log(p+d)$: indeed, the encoding of each phrase consists in a new
  symbol
  and a pointer towards one of the $p+d$ words of the dictionary. The only
  remaining step is thus to evaluate the number $p$ of new words in the
  dictionary.

 Let us order the words of the dictionary by increasing length and call $t_1$
  the total length of the first $l$ words (that is, the $l$ smallest words),
  $t_2$ the total length of the $l$ following words (that is, words of index
  between $l+1$ and $2l$ in the order), and so on: $t_k$ is the sum of the size
  of the words with index between $(k-1)l+1$ and $kl$. Since the sum of the size
  of all these words is equal to $|w|$, we have $$|w|=\sum_{k\geq 1}
  t_k.$$
  Furthermore, since for each $k$ there are at most $l$ new
  words of size $k$, the words taken into account in $t_k$ all have size at
  least $k$: hence $t_k\geq kl$.
 Thus we obtain
  $$|w|=\sum_{k\geq 1} t_k\geq \sum_{k=1}^{p/l}kl\geq \frac{p^2}{2l}.$$
  Hence $p$ satisfies
  $$\frac{p^2}{2l}\leq |w|,\mbox{ that is, }p\leq\sqrt{2l|w|}.$$
  The size of the compression of $w$ is $p\log
  (p+d)\leq \sqrt{2l|w|}\log (d+\sqrt{2l|w|})$,
    which ends the proof of Lemma  \ref{LZ is OK_finite}.

    \end{proofof}

\begin{proofof}{of Theorem \ref{t.lz.1}}
Let $A, c \in \N$ with $c\ge 7$. For each $i\in\N$, let
$R_i$ be a Kolmogorov random string with $|R_i|=i$ (i.e.
$K(R_i)>i\log |\Sigma|-A$ for $A$ the constant just fixed).
Let
\[S_n= R_1R_2^{2^c}R_3^{3^c}\ldots R_n^{n^c}\] ($R_n^{n^c}$
means $n^c$ copies of $R_n$) and let $S$ be the infinite
sequence having all $S_n$ as prefixes.

    The following three lemmas will analyze the performance of Lempel Ziv on all prefixes of $S$.
\begin{lemma}\label{l.lz.1}
    $$
    \frac{|LZ(S_n)|}{|S_n|}\leq\frac{n^{\frac{c+6}{2}}}{n^{c+1}}
    $$
    for $n$ large enough.
\end{lemma}

\begin{proofof}{of Lemma \ref{l.lz.1}}
    Denote by $LZ(i|i-1)$ the output of LZ on $R_i^{i^c}$,
    after having parsed $S_{i-1}$ already.

    Using the notation of Lemma \ref{LZ is OK_finite}, let $w=R_i^{i^c}$; thus $l=1+|R_i|=1+i$, and $d\leq |S_{i-1}|\leq (i-1)^{c+2}$. Thus
    $$
    |LZ(i|i-1)|\leq \sqrt{2(i+1)i^{c+1}} \log((i-1)^{c+2}+\sqrt{2(i+1)i^{c+1}})<i^{(c+3)/2}
    $$
    for $i$ large enough ($i\geq N_0$).
    Thus for $n$ sufficiently large
    \begin{align*}
    |LZ(S_n)|&= \sum_{j=1}^{n}|LZ(j|j-1)|\\
    &= \sum_{j=1}^{N_0-1}|LZ(j|j-1)|+\sum_{j=N_0}^{n}|LZ(j|j-1)|\\
    &\leq n + n\cdot n^{(c+3)/2} \leq n^{(c+6)/2}
    \end{align*}
    for $n$ large enough,
     which ends the proof of Lemma \ref{l.lz.1}.
    \end{proofof}

    \begin{lemma}\label{l.lz.2}
    Let $S_{n,t}= R_1R_2^{2^c}R_3^{3^c}\ldots R_n^{n^c}R^t_{n+1}$ where
    $1\leq t <(n+1)^c$.
    Then
    $$
    \frac{|LZ(S_{n,t})|}{|S_{n,t}|}\leq\frac{n^{(c+7)/2}}{n^{c+1}}
    $$
    for $n$ large enough.
\end{lemma}
 \begin{proofof}{of Lemma  \ref{l.lz.2}}

     Using Lemma \ref{l.lz.1} we have
    \begin{align*}
    |LZ(S_{n,t})|&= |LZ(S_n)|+ |LZ(R^t_{n+1}|S_n)|\\
    &\leq n^{(c+6)/2} + |LZ(R^t_{n+1}|S_n)|
    \end{align*}
    Applying Lemma \ref{LZ is OK_finite} with $w=R^t_{n+1}$, $d\leq |S_n| \leq n^{c+2}$, $l=n+2$, $|w|=t(n+1)$
    yields (for $n$ large enough)
    \begin{align*}
    |LZ(R^t_{n+1}|S_n)|&\leq \sqrt{2t(n+1)(n+2)} \log(n^{c+2}+\sqrt{2t(n+1)(n+2)})\\
    &\leq n^{3/2}\sqrt{t}\leq n^{(c+5)/2} .
    \end{align*}
    Whence
    $$
    \frac{|LZ(S_{n,t})|}{|S_{n,t}|}\leq\frac{n^{(c+6)/2}+n^{(c+5)/2}}{n^{c+1}} \leq \frac{n^{(c+7)/2}}{n^{c+1}}
    $$
    which ends the proof of Lemma \ref{l.lz.2}.
    \end{proofof}
    \begin{lemma}\label{l.lz.2a} For almost every $k$,
    $\frac{|LZ(S[1\ldots k]|)}{k}\leq k^{(-1+9/(c+3))/2}$
    i.e., for any $c\geq 7$
    $R_{LZ}(S) = 0.$
    \end{lemma}

     \begin{proofof}{of Lemma \ref{l.lz.2a}}
     Let $k\in\N$ and let $n,t,l$ ($0\leq l\leq
n$, $0\leq t<(n+1)^c$) be such that
    $S[1\ldots k]=S_nR^t_{n+1}R_{n+1}[1\ldots l]$.
    On $R_{n+1}[1\ldots l]$, LZ outputs at most $l \log(S[1\ldots k])= O(n \log n)$ symbols.
    Since $k\le (n+1)^{c+2}<n^{c+3}$,  Lemma \ref{l.lz.2}
    yields
    $$
    \frac{|LZ(S[1\ldots k])|}{k}\leq\frac{n^{(c+7)/2}+O(n\log n)}{n^{c+1}}\leq n^{(-c+6)/2} \leq k^{(-1+9/(c+3))/2} \ .
    $$ 
    \end{proofof}

    Let us show that the sequence $S$ is not compressible by ILplogs. For this we show that each large substring $x$ of the input that is a Kolmogorov
    random word cannot be compressed by a plogon transducer, independently of the computation performed before processing $x$.

    Let $C$ be an ILplog.
    For strings $z, \alpha, \beta,x$ with $z=\alpha x\beta$ and $|z|=m$, denote by $C(s,x,m)$ the output of $C$
    starting in configuration $s$ and reading $x$ out of an input of length $m$. A valid configuration,
    is a configuration $s$ such that there exists a string $c$ such that $C(s_0,c,m)$ ends in configuration
    $s$, where $s_0$ is the start configuration of $C$. For example
    if $s$ is the configuration of $C$ after reading $a$, then $C(s,x,m)$ is the output of
    $C$ while reading part $x$ of input $z=axb$. Note that $|s|\leq \log(m)^{{O}(1)}$.

    \begin{lemma} \label{l.lower.bound.plogon}
    Let $C$ be an ILplog, running in space $\log^a m$, and let $0<T\leq 1$.
    Then for every $d\in\N$ and almost every $r\in\N$,
    for every random string $x\in\Sigma^{r}$
    (with $K(x)\geq T|x|\log |\Sigma|-A$ for some fixed constant $A$),
    for every $M$ with $|x|\leq M \leq |x|^d$
    and for every  valid configuration $s$ ($|s|\leq\log^a M$)
    $$
    |C(s,x,M)|\geq T|x| - \log^{2a}|x|.
    $$
    \end{lemma}
    \begin{proofof}{of Lemma \ref{l.lower.bound.plogon}}
    Suppose by contradiction that $C(s,x,M)=p$, with $|p|< Tr - \log^{2a}r$;
    denote by $s^x$ the configuration of $C$ after having read $x$
    starting in $s$.
    Then $p'=(s^x,s,M,r,p)$ ($p'$ is encoded by doubling all symbols in $s^x,s,M,r$,
    separated by the delimiter $01$ followed by $p$) yields a program for $x$  (coded with alphabet $\Sigma$):\\
    ``Find  $y$ with $|y|=r$ such that $C(s,y,M)=p$, and $C$ ends in configuration $s^x$ after
    reading $y$.''

    $y$ is unique because otherwise suppose there are two strings $y,y'$ ($|y|=|y'|$) such that
    $C(s,y,M)=C(s,y',M)$, and  $C$ ends in the same configuration on $y$ and $y'$.
    Let $b$ be a string that brings $C$ into configuration $s$. Then
    for $z=1^{M-|by|}$ we have $C(byz)=C(by'z)$ which contradicts $C$ being 1-1.
    Therefore $y$ is unique, i.e. $y=x$.
    Thus for $r$ sufficiently large
    $$|p'|\leq 2(|s^x|+|s|+|M|+|r|)+|p| \leq 2(\log^a r^d + \log^a r^d +\log r^d + \log r) + Tr - \log^{2a}r$$
    $$\leq Tr - \frac{\log^{2a}r}{2} $$
    which contradicts the randomness of $x$. 
    \end{proofof}

    \begin{lemma}\label{otrolemma}
    Let $C$ be an ILplog, running in space $\log^a m$.
    Then for every $\epsilon>0$ and for almost every $m$,
    $
    \frac{C(S[1\cdots m])}{m} > 1-\epsilon
    $
    i.e., $\rho_{\pl}(S)=1.$
    \end{lemma}

\begin{proofof}{of Lemma \ref{otrolemma}}
    Let $\epsilon>0$ and let $\epsilon'=\frac{\epsilon}{4\cdot 3^{c+2}}$.
    Let $n,t,l$ ($0\leq l\leq n$, $0\leq t<n^c$) be such that
    $S[1\ldots m]=S_{n-1}R^t_{n}R_{n}[1\ldots l]$.

    The idea is to apply Lemma \ref{l.lower.bound.plogon} to
    $R_{\epsilon'n}^{(\epsilon'n)^c} \ldots R_{n-1}^{(n-1)^c}R^t_{n}R_{n}[1\ldots l]$.
    Let $d$ be such that $(\epsilon'n)^d\geq n^{c+2}$ (for all $n\geq 2$), i.e.
    $(\epsilon'n)^d\geq m$.
    By Lemma \ref{l.lower.bound.plogon}, $C$ on input $S[1\ldots m]$,
    will output at least $j-\log^{2a} j$ symbols on each $R_j$ ($\epsilon' n \leq j\le n$).
    Therefore
    $$
    |C(S[1\ldots m])|\geq \sum_{j=\epsilon'
    n}^{n-1}(j-\log^{2a}j)j^c+ t (n-\log^{2a}n)
    $$
    whence
    \begin{align*}
    \frac{|C(S[1\ldots m])|}{m}
    &\geq \frac{\sum_{j=\epsilon' n}^{n-1}j^c(j-\log^{2a}j)+ t (n-\log^{2a}n)}{\sum_{j=1}^{n-1}j^{c+1}+(t+1)n}
    \ \geq \frac{\sum_{j=\epsilon' n}^{n-1}j^c(j-\alpha j)+ t (n-\log^{2a}n)}{\sum_{j=1}^{n-1}j^{c+1}+(t+1)n}\\
    &\geq \frac{(1-\alpha)(\sum_{j=\epsilon'
    n}^{n-1}j^{c+1}+(t+1)n)}{(1+\alpha')(\sum_{j=1}^{n-1}j^{c+1}+(t+1)n)}-
    \frac{(1-\alpha)n}{\sum_{j=1}^{n-1}j^{c+1}+(t+1)n}
    \end{align*}
    where $\alpha,\alpha'>0$ can be chosen arbitrarily small (for $n$ large enough).
    Let $\alpha,\alpha'>0$ be such that $\frac{1-\alpha}{1+\alpha'}>1-\epsilon/2$.
    Thus
    \begin{align*}
    \frac{|C(S[1\ldots m])|}{m}
    &\geq \frac{1-\alpha}{1+\alpha'}
    -\frac{1-\alpha}{1+\alpha'}\cdot \frac{\sum_{j=1}^{\epsilon' n-1}j^{c+1}}
    {\sum_{j=1}^{n-1}j^{c+1}}- \epsilon/4
    \geq \frac{1-\alpha}{1+\alpha'}
    -\frac{\epsilon' n^{c+2}}
    {n/3 (n/3)^{c+1}}- \epsilon/4\\
    &= \frac{1-\alpha}{1-\alpha'}
    -\epsilon' 3^{c+2}- \epsilon/4
    > 1-\epsilon/2 - \epsilon/4 - \epsilon/4\\
    &> 1 -\epsilon \ .
    \end{align*}
    Since $\epsilon$ is arbitrary, $\rho_{\pl}(S)=1.$
\end{proofof}

This finishes the proof of Theorem \ref{t.lz.1}.
\end{proofof}

\subsection{Invertible pushdown beats plogon compressors}\label{secinvpd}

In this section we take the most restrictive classes of
pushdown compressors, namely invertible pushdown automata
and visibly pushdown automata, and show that they both
outperform plogon compressors.

The proof is based on using a list of Kolmogorov random
strings together with their reverses to construct the
sequence witnessing the separation. A careful choice of the
length of these random strings makes the result
incompressible by plogon devices.

\begin{theorem}\label{theo:ipd}
    For each $\epsilon>0$ there exists a sequence $S$ such that
    $$R_{\ipd}(S) \leq 1/2 \quad \text{and }\rho_{\pl}(S)\ge 1-\epsilon.$$
\end{theorem}

\begin{proof}
    Let $\epsilon_1, \epsilon_2>0$ and let $k\in\N$ to be determined later (as $k>4/\epsilon_2$).

    We first notice that for each $m\in\N$ there is a string $y\in\Sigma^*$ with $|y|=k m$ and such that
    $y[ik+1 .. (i+1)k]\ne 1^k$ for every $i$ and $K(y)\ge \frac{k-1}{k}|y|\log |\Sigma|$. This can be proved by a simple counting argument.

    Let $t_n=k^{\lceil\frac{\log n}{\log k}\rceil}$, so that
    \begin{equation}\label{e.nk}
    n \leq t_n \leq nk.
    \end{equation}
    For each $n\in\N$ let $y_n\in\Sigma^{kt_n}$ be as above ($y_n[ik+1 .. (i+1)k]\ne 1^k$ for every $i$ and $K(y_n)\ge \frac{k-1}{k}|y_n|\log |\Sigma|$).

    Consider the sequence $S=y_11^ky_1^{-1}y_2 1^k y_2^{-1}\ldots y_n 1^k y_n^{-1}\ldots$. We will refer to the $1^k$ separators as flags.
    Consider the following invertible pushdown compressor $(C,D)$.
    Informally  on both $y_j$ and flag zones,  $C$ outputs the input.
    On a  $y_j^{-1}$ zone, $C$ outputs a zero for every $1/\epsilon_1$ symbols,
    and checks using the stack that the input is indeed $y_j^{-1}$. If the test fails,
    $C$ outputs an error flag, enters an error state, and from then on it outputs the input.

    The complete definitions of $C$ and $D$ are given for the sake of completeness. Let $A\geq 1/\epsilon_1$ with $A=k^a$ for some $a\in\N$, i.e.
    guaranteeing that $A| \, |y_n|$ for almost every $n$.
    The set of states $Q$ is:
    \begin{itemize}
        \item the start state $q^s_0$
        \item the counting states $q^s_1,\ldots,q^s_b$ and $q_0$, with $b=k\sum_{j=1}^{2^{\lceil a\log k\rceil}}(2t_j+1)$
        \item the flag  checking states $q_1^{f_1},\ldots,q_k^{f_1}$ and $q_1^{f_0},\ldots,q_k^{f_0}$
        \item the pop flag states $q^r_0,\ldots, q^r_k$
        \item the compress states $q^c_1,\ldots, q^c_{A+1}$
        \item the error state $q^e$.
    \end{itemize}
    We now describe the transition function $\delta:Q\times\Sigma^{*}\times\Sigma^{*}\rightarrow Q\times\Sigma^{*}$.
    At first $C$ counts from $q^s_0$ to $q^s_b$. This guarantees that  for later $y_j$, $A| \, |y_j|$.
    For $0\leq i \leq b-1$ let
    $$\delta(q^s_i,x,y) =(q^s_{i+1},y)$$ and
    $$\delta(q^s_b,\lambda,y) =(q_0,y) .$$

    After counting has taken place, a new $y$ zone starts;
    the input is pushed to the stack, and it is checked for the flag, by groups of $k$ symbols.
    $$
            \delta(q_0,x,y) =
        \begin{cases}
        (q^{f_1}_1,xy) &\text{ if } x=1\\
        (q^{f_0}_1,xy) &\text{ if } x\ne 1\\
        \end{cases}
   $$
    and for $1\leq i \leq k-1$
    $$\delta(q^{f_0}_i,x,y)=(q^{f_0}_{i+1},xy)$$
    $$
            \delta(q^{f_1}_i,x,y) =
        \begin{cases}
        (q^{f_1}_{i+1},xy) &\text{ if } x=1\\
        (q^{f_0}_{i+1},xy) &\text{ if } x\ne 1\\
        \end{cases}
   $$
    If the flag has not been detected after $k$ symbols, the  test starts again.
    $$\delta(q^{f_0}_k,\lambda,y)=(q_0,y).$$
    If the flag has  been detected the pop flag state is entered
    $$\delta(q^{f_1}_k,\lambda,y)=(q^r_0,y).$$
    Since the flag has been pushed to the stack it has to be removed, thus for
    $0\leq i \leq k-1$
    $$\delta(q^r_i,\lambda,y)=(q^r_{i+1},\lambda)$$
    $$\delta(q^r_k,\lambda,y)=(q^c_1,y).$$
    $C$ then checks using the stack that the input is indeed $y_j^{-1}$, counting modulo $A$. If the test fails,
    an error state is entered, thus for
    $1\leq i \leq A$
    $$
            \delta(q^{c}_i,x,y) =
        \begin{cases}
        (q^{c}_{i+1},\lambda) &\text{ if } x=y\\
        (q^{e},y) &\text{ if } x\neq y\text{ and }y\ne z_0\\
            (q^{f_1}_{1},x z_0) &\text{ if } x=1,\ y= z_0\\
            (q^{f_0}_{1},x z_0) &\text{ if } \ne 1,\ y= z_0\\
        \end{cases}
    $$
    Once $A$ symbols have been checked, the test starts again
    $$\delta(q^c_{A+1},\lambda,y)=(q^c_1,y).$$
    The error state is a loop,
    $\delta(q^e,x,y)=(q^e,y)$.

    We next describe the output function $\nu:Q\times\Sigma^{*}\times\Sigma^{*}\rightarrow \Sigma^{*}$.
    First on the counting states, the input is output, i.e., for $0\leq i\leq b-1$
    $$\nu(q^s_i,x,y)=x.$$

    On the flag states the input is output, thus for $1\leq i \leq k-1$, $a\in\bool{}$
    $$\nu(q^{f_a}_{i},x,y)=x.$$
    There is no output on popping states $q^r_0,\ldots, q^r_k$ and on compressing states $q^c_1,\ldots, q^c_{A+1}$
    except after $A$ symbols have been checked i.e.
    $$\nu(q^{c}_{A},x,y)=0\text{ if }x=y$$
    On error, $1^i0x$ is output, i.e. for $1\leq i \leq A$
    $$\nu(q^{c}_{i},x,y)=1^{i}0x\text{ if } x\neq y\text{ and }y\ne z_0.$$
    On the error state, the input is output, that is,
    $\nu(q^{e},x,y)=x$.

    Let us verify $C$ is IL, that is, the input can be recovered from the output and the final state. If the final state is not an error  state, then both all $y_j$'s and
    all flags are output as in the input. If the final state is  $q^c_i$ then the number $t$ of zeroes after the last flag (in the output),
    together with the final state $q^c_i$ determines that the last $y_j^{-1}$ zone is $tA+i-1$ symbols long.

    If the final state is  an error  state, then the output is of the form (suppose the error
    happened in the $y_j^{-1}$ zone)
    $$ay_j1^k0^t1^{i}0b$$
    with $a,b\in\Sigma^{*}$. The input is uniquely determined to be the input corresponding to output
    $ay_j1^k0^t$ with final state $q_1^c$ followed by
$$y_j^{-1}[tA+1 .. tA+i-1] b.$$

    We give the definition of the inverter $D$.
    The set of states $Q'$ is:
    \begin{itemize}
        \item the start state $q^s_0$
        \item the counting states $q^s_1,\ldots,q^s_b,q_0$, with $b=k\sum_{j=1}^{2^{\lceil a\log k\rceil}}(2t_j+1)$
        \item the flag  checking states $q_1^{f_1},\ldots,q_k^{f_1}$ and $q_1^{f_0},\ldots,q_k^{f_0}$
        \item the pop flag states $q^r_0,\ldots, q^r_k$
        \item the decompress states $q^d_u$ for $u\in\Sigma^{\le A}$
        \item the copy states $q^w_u$ for $u\in\Sigma^{\le A}$
        \item the output state $q^o$
    \end{itemize}
    $D$ receives as input a string followed by a state $q_f\in Q$.
    Let us describe the transition function $\delta':Q'\times\Sigma^{*}\times\Sigma^{*}\rightarrow Q'\times\Sigma^{*}$ and
    the output function $\nu':Q'\times\Sigma^{*}\times\Sigma^{*}\rightarrow \Sigma^{*}$ in parallel.
    At first $D$ counts from $q^s_0$ to $q^s_b$, i.e.,
    for $0\leq i \leq b-1$ let
    $$\delta'(q^s_i,x,y) =(q^s_{i+1},y)$$ and
    $$\delta'(q^s_b,\lambda,y) =(q_0,y) .$$
    On the counting states, the input is output, i.e., for $0\leq i\leq b-1$
    $$\nu'(q^s_i,x,y)=x.$$

    At first the input is pushed to the stack, and it is checked for the flag, by groups of $k$ symbols.
    $$
            \delta'(q_0,x,y) =
        \begin{cases}
        (q^{f_1}_1,xy) &\text{ if } x=1\\
        (q^{f_0}_1,xy) &\text{ if } x\ne 1\\
        \end{cases}
   $$
    and for $1\leq i \leq k-1$
    $$\delta'(q^{f_0}_i,x,y)=(q^{f_0}_{i+1},xy)$$
    $$
            \delta'(q^{f_1}_i,x,y) =
        \begin{cases}
        (q^{f_1}_{i+1},xy) &\text{ if } x=1\\
        (q^{f_0}_{i+1},xy) &\text{ if } \ne 1\\
        \end{cases}
   $$
    If the flag has not been detected after $k$ symbols, the  test starts again.
    $$\delta'(q^{f_0}_k,\lambda,y)=(q_0,y).$$
    If the flag has  been detected the pop flag state is entered
    $$\delta'(q^{f_1}_k,\lambda,y)=(q^r_0,y).$$
    Since the flag has been pushed to the stack it has to be removed, thus for
    $0\leq i \leq k-1$
    $$\delta'(q^r_i,\lambda,y)=(q^r_{i+1},\lambda)$$
    $$
    \delta'(q^r_k,\lambda,y)=(q^d_{\lambda},y)
    $$
    On the flag states the input is output, i.e.
    for $1\leq i \leq k-1$, $a\in\bool{}$
    $$\nu'(q^{f_a}_{i},x,y)=x,$$
    $$\nu'(q_0,x,y)=x.$$
    There is no output on popping states $q^r_0,\ldots, q^r_k$.

The decompressing states pop and memorize $A$ symbols of
the stack
$$
    \delta'(q^d_u,\lambda,y)=(q^d_{uy},\lambda)\text{ for }|u|<A.
    $$

    If $|u|=A$ then, depending on the next symbol, $u^{-1}$ should be output
    $$
    \delta'(q^d_u,0,y)=(q^d_{\lambda},y)\text{ if }y\ne z_0.
    $$
        $$
    \delta'(q^d_u,0,z_0)=(q_0,z_0).
    $$
    $$
    \nu'(q^d_u,0,y)=u^{-1}.
    $$
    If 1 is found then there is an error
    $$
    \delta'(q^d_u,1,y)=(q^w_u,y).
    $$
    $$
    \delta'(q^{w}_{bu},1,y)=(q^w_u,y).
    $$
      $$
    \nu'(q^w_{bu},1,y)=b.
    $$
     $$
    \delta'(q^{w}_{bu},0,y)=(q^o,y).
    $$

 If the next symbol is a state then the $y^{-1}$ zone was not complete
 $$
    \nu'(q^d_{u},q_i^c,y)=u^{-1}[1..i-1].
    $$

    Once the error has been passed, $D$ stays in the output state.
    $\delta'(q^o,x,y)=(q^o,y)$, $\nu'(q^o,x,y)=x$.

    This ends the description of $(C,D)$.

    Let us compute the compression ratio of $C$. For $n$ large enough and since the counting part on the first $b$ symbols of $S$
    is of constant size, it is negligible for computing the compression ratio, therefore we can assume wlog that $C$ starts compressing
    immediately, i.e. $b=0$; moreover  the ratio is largest just after a flag $1^k$
    whence
    \begin{align*}
    \frac{|C(y_1 1^k y_1^{-1}y_2 1^k y_2^{-1} \ldots y_n1^k)|}
    {|y_1 1^k y_1^{-1}y_2 1^k y_2^{-1} \ldots y_n1^k|}
    &\leq \frac{k(1+\epsilon_1)\sum_{j=1}^{n}t_j+nk-\epsilon_1 kt_n}{2k\sum_{j=1}^{n}t_j+nk-kt_n}\\
    &\leq \frac{1+\epsilon_1}{2} + \frac{n/2}{\sum_{j=1}^{n-1}t_j} + \frac{t_n/2}{\sum_{j=1}^{n-1}t_j}\\
    &\leq 1/2 + \epsilon_1/2+ \frac{n}{n(n-1)} + \frac{nk}{n(n-1)}\\
    &<1/2 + \epsilon_1/2 +\epsilon_1/4 +\epsilon_1/4 = 1/2 +\epsilon_1
    \end{align*}
    for $n$ sufficiently large. Since $\epsilon_1$ is arbitrary $$R_{\ipd}(S)\leq 1/2.$$

We now compute the compression ratio of a plogon compressor
on $S$.
    Let $m\in\N$ and let $n\in\N$ be such that
    $$S[1\ldots m] = y_1 1^k y_1^{-1}y_2 1^k y_2^{-1} \ldots (y_n1^ky_n^{-1})[1\ldots i]$$
    with $1\leq i \leq k(1+2t_n)$. Let $C$ be an ILplog, running in space $\log^a m$.
    Let $\epsilon'=\epsilon_2/8k$. Applying Lemma \ref{l.lower.bound.plogon} with $d=3$ and $r$ ranging  $\epsilon'n\leq r\leq n$
    (such that $r\le m\le r^3$ for $n$ sufficiently large),
    we have that for every $j\in\{ \epsilon' n,\ldots ,n\}$
    $$
    |C(s,y_j^{\delta},m)|\geq T|y_j| - \log^{2a}(|y_j|)
    $$
    where $\delta=\pm 1$.
    Letting $s_j$ (resp. $s_j'$) ($j\in\{ \epsilon' n,\ldots ,n\}$)
    denote the configuration of $C$ reached on input $S[1\ldots m]$
    just before reading the first symbol of $y_j$ (resp. $y_j^{-1}$), we have
    \begin{align*}
    |C(S[1\ldots m])|&\geq \sum_{j=\epsilon' n}^{n-1}|C(s_j,y_j,m)|+\sum_{j=\epsilon' n}^{n-1}|C(s'_j,y^{-1}_j,m)|\\
    &\geq 2  \sum_{j=\epsilon' n}^{n-1}(T|y_j|-\log^{2a}|y_j|)\\
    &> 2  \sum_{j=\epsilon' n}^{n-1}(T|y_j|-\gamma |y_j|)\\
    &= 2(T-\gamma)  \sum_{j=\epsilon' n}^{n-1}|y_j|
    \end{align*}
    with $\gamma>0$ arbitrary close to $0$, for $n$ large enough.
    Choosing $\gamma$ and $T=\frac{k-1}{k}$ such that $T-\gamma>1-\epsilon_2/4$ (taking $k>4/\epsilon_2$)
    yields
    \begin{align*}
    \frac{|C(S[1\ldots m])|}{|S[1\ldots m]|}
    &\geq \frac{2(T-\gamma)  \sum_{j=\epsilon' n}^{n-1}kt_j}{nk+2\sum_{j=1}^{n}kt_j}\\
    &\geq (T-\gamma) -(T-\gamma)[\frac{n/2}{\sum_{j=1}^{n}t_j} + \frac{t_n}{\sum_{j=1}^{n}t_j}+\frac{\sum_{j=1}^{\epsilon'n-1}t_j}{\sum_{j=1}^{n}t_j}]\\
    &\geq (T-\gamma) -(T-\gamma)[\frac{n}{n(n-1)} + \frac{2kn}{n(n-1)}+ \frac{k\epsilon'n(\epsilon'n-1)}{n(n-1)}]\\
    &\geq 1-\epsilon_2/4 - \epsilon_2/4 - \epsilon_2/4 -\epsilon_2/4 \\
    &> 1 -\epsilon_2
    \end{align*}
    for $n$ sufficiently large, and
    $$\rho_{\pl}(S)\geq 1-\epsilon_2.$$
\end{proof}

Even visibly pushdown automata, extensively used in the
compression of XML, can beat plogon compressors. The
definition of visibly pushdown automata can be found in
section \ref{s.pushdown.compressors}.

\begin{theorem}\label{theo:visi}
    There exists a sequence $S$ such that
    $$R_{\mathrm{visiblyPD}}(S) \leq 1/2 \quad \text{and } \rho_{\pl}(S)\ge 1-\frac{1}{\log |\Sigma|}.$$
\end{theorem}

\begin{proof}

The proof is a variation of the proof of Theorem
\ref{theo:ipd}. If the alphabet $\Sigma$ has $2t$ symbols,
this time the sequence used is $S=y_1Y_1^{-1}y_2
Y_2^{-1}\ldots y_n  Y_n^{-1}\ldots$, where $y_i$ are
Kolmogorov random strings over the first $t$ symbols of the
alphabet, and $Y_i$ is the string obtained from $y_i$ by
changing each symbol $a$ by symbol $a+t$, that is, $Y_i$
contains only the last $t$ symbols of the alphabet.
\end{proof}

\subsection{Lempel-Ziv is not universal for
Pushdown compressors}\label{sec_LZnouniversalPD} It is well
known that LZ \cite{LemZiv78} yields a lower bound on the
finite-state compression of a sequence \cite{LemZiv78},
i.e., LZ is universal for finite-state compressors.

The following result shows that this is not true for
pushdown compression, in a strong sense: we construct a
sequence $S$ that is infinitely often incompressible by LZ,
but that has almost everywhere pushdown compression ratio
less than $\frac{1}{2}$.

\begin{theorem} \label{th LZ_not_univ_PD}
For every $\epsilon>0$, there is a sequence $S$ such that
$$R_{\ipd}(S)\leq \frac{1}{2}$$
and $$\rho_{LZ}(S)
> 1 - \epsilon.$$
\end{theorem}

\begin{proof}
 Let $\epsilon>0$, and let
    $k=k(\epsilon), v=v(\epsilon), v'=v'(\epsilon)$ be  integers to be determined later.
        For any integer $n$, let $T_n$ denote the set of  strings $x$ of size $n$ such that
        $1^j$ does not appear in $x$, for every $j\geq k$.
        Since $T_n$ contains $\Sigma^{k-1}\times \{ 0 \} \times \Sigma^{k-1}\times \{ 0 \} \ldots $
        (i.e. the set of strings whose every $k$th symbol is zero),
        it follows that
        $|T_n|\geq |\Sigma|^{an}$, where $a=1-1/k$.

        \begin{remark} \label{r.extension}
            For every string $x\in T_n$ there is a string $y\in T_{n-1}$ and a symbol $b$ such that $yb=x$.
        \end{remark}

        Let $A_n = \{a_1,\ldots a_u\}$ be the set of palindromes in $T_n$. Since fixing the $n/2$ first symbols of a palindrome (wlog $n$ is even)
        completely determines it, it follows that $|A_n| \leq |\Sigma|^{\frac{n}{2}}$.
        Let us separate the remaining strings in $T_n-A_n$ into $v$ pairs of sets
        $X_{n,i} = \{x_{i,1},\ldots x_{i,t}\}$ and $Y_{n,i} = \{y_{i,1},\ldots y_{i,t}\}$ with
    $t=\frac{|T_n-A_n|}{2v}$,
        $(x_{i,j})^{-1}=y_{i,j}$ for every $1\leq j \leq t$ and $1\leq i \leq v$,
    $x_{i,1}, y_{i,t}$ start with a zero. For convenience we write $X_{i}$ for $X_{n,i}$.

        We construct $S$ in stages.
    Let $f(k)=2k$ and $f(n+1)=f(n)+v+1$. Clearly $$ n^2>f(n)>n.$$ For $n \leq k-1$,
        $S_n$ is an enumeration of all strings of size $n$ in lexicographical order.
        For $n\geq k$,
        \begin{align*}
            S_n = \ &a_1 \ldots a_u \  1^{f(n)} \ x_{1,1} \ldots x_{1,t} \ 1^{f(n)+1} \ y_{1,t} \ldots y_{1,1}
    \ x_{2,1} \ldots x_{2,t} \ 1^{f(n)+2} \ y_{2,t} \ldots y_{2,1} \ldots \\
    &\ldots x_{v,1} \ldots x_{v,t} 1^{f(n)+v} y_{v,t} \ldots y_{v,1}
    \end{align*}
        i.e. a concatenation of all strings in $A_n$ (the $A$ zone of $S_n$) followed by a flag of $f(n)$ ones,
        followed by the concatenations of all strings in the  $X_i$ zones and $Y_i$ zones, separated
    by flags of increasing length. Note that the $Y_i$ zone is exactly the $X_i$ zone written in reverse
    order.
    \noindent
        Let $$S=S_1 S_2 \ldots S_{k-1} \ 1^k \ 1^{k+1} \ \ldots 1^{2k-1} \ S_{k} S_{k+1} \ldots $$
        i.e. the concatenation of the $S_j$'s with some extra flags between $S_{k-1}$ and $S_k$.
        We claim that the parsing of $S_n$ ($n\geq k$) by LZ, is as follows:
        $$a_1, \ldots , a_u, \  1^{f(n)}, \ x_{1,1}, \ldots, x_{1,t}, \ 1^{f(n)+1}, \ y_{1,t}, \ldots, y_{1,1},
    \ldots,
    x_{v,1}, \ldots, x_{v,t}, 1^{f(n)+v}, y_{v,t}, \ldots, y_{v,1} .$$
        Indeed after $S_1, \ldots S_{k-1} \ 1^k \ 1^{k+1} \ \ldots 1^{2k-1}$, LZ has parsed every
        string of size $\leq k-1$ and the flags $1^k \ 1^{k+1} \ \ldots 1^{2k-1}$. Together with Remark \ref{r.extension},
        this guarantees that LZ parses $S_n$ into phrases that are exactly all the strings in $T_n$ and
        the $v+1$ flags $1^{f(n)},\ldots,1^{f(n)+v}$.

        Let us compute the compression ratio $\rho_{LZ} (S)$.
        Let $n,i$ be integers. By construction of $S$, LZ encodes every phrase in $S_i$ (except flags),
          by a phrase in $S_{i-1}$ plus one symbol.
        Indexing a phrase in $S_{i-1}$ requires a codeword of length at least logarithmic in the number of phrases parsed
        before, i.e. $\log (P(S_1 S_2 \ldots S_{i-2}))$.
        Since $P(S_i)\geq |T_i| \geq |\Sigma|^{ai}$, it follows that for almost every $i$
        $$
        P(S_1 \ldots S_{i-2}) \geq \sum^{i-2}_{j=1}|\Sigma|^{aj} = \frac{|\Sigma|^{a(i-1)} -|\Sigma|^a}{|\Sigma|^a-1} \geq b |\Sigma|^{a(i-1)}
        $$
        where the inequality holds because $a<1$ (hence the denominator is less than 1).
        Letting $t_i=|T_i|$, the number of symbols output by LZ on $S_i$ is at least
        \begin{align*}
        P(S_i) \log P(S_1\ldots S_{i-2})
        &\geq t_i \log b |\Sigma|^{a(i-1)}\\
        &\geq c t_i(i-1)
        \end{align*}
        where $c=c(a)$ can be made arbitrarily close to $1$, by choosing $a$ accordingly.
        Therefore
        $$
        |LZ(S_1 \ldots S_n)| \geq  \sum_{j=1}^{n} c t_j(j-1)
        $$
        Since
    $$|S_1 \ldots S_n| = |S_1 \ldots S_{k-1} 1 \ldots 1|+|S_{k}\ldots S_{n}|
    \leq |\Sigma|^{3k} + \sum_{j=k}^{n} (j t_j + (v+1)(f(j)+v))  $$
    and
    $|LZ(S_1\ldots S_n)| \geq \sum_{j=k}^{n}ct_j(j-1)$,
    the compression ratio is given by

\begin{align*}
        \rho_{LZ}(S_1 \ldots S_n)   &\geq
    c\frac{\sum_{j=k}^{n} t_j(j-1) }{|\Sigma|^{3k} + \sum_{j=k}^{n} (j t_j + (v+1)(f(j)+v))}\\
        &= c - c \frac{|\Sigma|^{3k} + \sum_{j=k}^{n} (j t_j + (v+1)(f(j)+v) -t_j(j-1))  }
    {|\Sigma|^{3k} + \sum_{j=k}^{n} (j t_j + (v+1)(f(j)+v))} \\
    &= c - c\frac{|\Sigma|^{3k} + \sum_{j=k}^{n} (t_j + (v+1)(f(j)+v))  }
    {|\Sigma|^{3k} + \sum_{j=k}^{n} (j t_j + (v+1)(f(j)+v))  }
        \end{align*}

        The second term in this equation 
        can be made arbitrarily small for $n$ large enough:
        Let $k < M\leq n/3$, we have
        \begin{align*}
        \sum_{j=k}^{n} j t_j &\geq \sum_{j=k}^{M} j t_j + (M+1)\sum_{j=M+1}^{n}  t_j\\
        &= \sum_{j=k}^{M} j t_j + M\sum_{j=M+1}^{n} t_j+ \sum_{j=M+1}^{n}  t_j\\
          &\geq \sum_{j=k}^{M} j t_j + M\sum_{j=M+1}^{n} t_j+ \sum_{j=M+1}^{n}  |\Sigma|^{aj}\\
          &\geq  \sum_{j=k}^{M} j t_j + M\sum_{j=M+1}^{n} t_j+ |\Sigma|^{an}\\
        \end{align*}
    We have
    $$
    |\Sigma|^{an} \geq M [ |\Sigma|^{3k} +\sum_{j=k}^{M} t_j + (v+1) \sum_{j=k}^{n} (f(j)+v) ]
    $$
    for $n$ large enough, because $f(j)<j^2$.
    Hence
    $$
    c\frac{|\Sigma|^{3k} + \sum_{j=k}^{n} (t_j + (v+1)(f(j)+v))  }
    {|\Sigma|^{3k} + \sum_{j=k}^{n} (j t_j + (v+1)(f(j)+v))  }
    \leq
    c\frac{|\Sigma|^{3k} + \sum_{j=k}^{n} (t_j + (v+1)f(j)+v)  }
    {M[ |\Sigma|^{3k} + \sum_{j=k}^{n} (t_j + (v+1)(f(j)+v))  ]} = \frac{c}{M}
    $$
    i.e.
        $$
        \rho_{LZ}(S_1 \ldots S_n)   \geq c - \frac{c}{M}
        $$
        which by definition of $c, M$ can be made arbitrarily close to $1$ by choosing $k$ accordingly, i.e
        $$
        \rho_{LZ}(S_1 \ldots S_n)   \geq 1- \epsilon.
        $$

        Let us show that $R_{PD}(S)\leq \frac{1}{2}$. Consider the following ILPD compressor $C$.
          First $C$ outputs its input until it reaches zone $S_k$.
        Then on any of the zones $A,X_i$ and the flags, $C$ outputs them symbol by symbol; on $Y_i$ zones, $C$
          outputs one zero for every $v'$ symbols of input. To recognize a flag: as soon as $C$ has read $k$ ones, it knows it has reached a flag.
            For the stack: $C$ on  $S_n$ cruises through the $A$ zone up to the
            first flag, then starts pushing the whole $X_1$
            zone onto its stack until it hits the second flag. On $Y_1$, $C$ outputs a $0$ for every $v'$ symbols of input,
            pops one symbol from the stack for every symbol of input, and cruises through $v'$ counting states, until the stack is
            empty (i.e. $X_2$ starts).
            $C$  keeps doing the same for each pair $X_i,Y_i$ for every
            $2\leq i \leq v$. Therefore at any time, the number of symbols of $Y_i$ read so far is equal to
            $v'$ times the number of symbols output on the $Y_i$ zone plus the index of the current counting state.
            On the $Y_i$ zones, $C$ checks that every symbol of $Y_i$ is equal to the symbol it pops from the stack;
            if the test fails, $C$ enters an error state, outputs an error flag and thereafter outputs every symbol it reads  (this guarantees
            IL on sequences different from $S$).
            This together with the fact that the $Y_i$ zone is exactly the $X_i$ zone written in reverse order,
            guarantees that $C$ is IL. Before giving a detailed
        construction of $C$, we compute the upper bound it yields on $R_{PD}(S)$.

    \begin{remark}\label{r.martingale.capital}
        For any $j\in\N$, let $p_j=C(S[1\ldots j])$ be the output of $C$ after reading $j$ symbols of $S$.
        Is it easy to see that the ratio $\frac{|p_j|}{|S[1\ldots j]|}$ is maximal
        at the end of a flag following an $X_i$ zone, since the flag is followed by a $Y_i$ zone,
        on which $C$ outputs one symbol for every $v'$ input symbols.
    \end{remark}

    Let $0\leq I < v$. We compute the ratio $\frac{|p_j|}{|S[1\ldots j]|}$ inside zone $S_n$
    on the last symbol of the flag following $X_{I+1}$. At this location (denoted $j_0$), $C$ has output
    \begin{align*}
    |p_{j_0}| &\leq |\Sigma|^{3k} + \sum_{j=k}^{n-1}[j|A_j| + (v+1)(f(j)+v)+\frac{j}{2}|T_j-A_j|(1+\frac{1}{v'})]
                + n|A_n|+(v+1)(f(n)+v) \\&+ \frac{n}{2v}|T_n-A_n|(I+1+\frac{I}{v'})\\
                &\leq |\Sigma|^{pn} + \sum_{j=k}^{n-1}[\frac{j}{2}|T_j|(1+\frac{1}{v'})]
                + \frac{n}{2v}|T_n|(I+1+\frac{I}{v'})
    \end{align*}
    where $p>\frac{1}{2}$ can be made arbitrarily close to $\frac{1}{2}$ for $n$ large enough.

    The number of symbols of $S$ at this point is
    \begin{align*}
        |S[1\ldots j_0]| &\geq \sum_{j=k}^{n-1}j|T_j|
                + n|A_n| + \frac{n}{v}|T_n-A_n|(I+\frac{1}{2})\\
                &\geq  \sum_{j=k}^{n-1}j|T_j|
                + \frac{n}{v}|T_n|(I+\frac{1}{4})
    \end{align*}

    Hence by Remark \ref{r.martingale.capital}
    \begin{align*}
        \limsup_{n\rightarrow\infty} \frac{|p_n|}{|S[1\ldots n]|}
      &\leq \limsup_{n\rightarrow\infty} \frac{|\Sigma|^{pn} + \sum_{j=k}^{n-1}[\frac{j}{2}|T_j|(1+\frac{1}{v'})]
                + \frac{n}{2v}|T_n|(I+1+\frac{I}{v'})}{\sum_{j=k}^{n-1}j|T_j|+ \frac{n}{v}|T_n|(I+\frac{1}{4})}\\
       &= \limsup_{n\rightarrow\infty}
        [
        \frac{|\Sigma|^{pn}}{\sum_{j=k}^{n-1}j|T_j|+ \frac{n}{v}|T_n|(I+\frac{1}{4})}
        + \frac{1}{2}\frac{\sum_{j=k}^{n-1}j|T_j| + \frac{n}{v}|T_n|(I+\frac{1}{4})}
            {\sum_{j=k}^{n-1}j|T_j|+ \frac{n}{v}|T_n|(I+\frac{1}{4})}
        \\&+ \frac{1}{2v'}\frac{\sum_{j=k}^{n-1}j|T_j|}{\sum_{j=k}^{n-1}j|T_j|+ \frac{n}{v}|T_n|(I+\frac{1}{4})}
        + \frac{n|T_n|}{2v}\frac{\frac{I}{v'}+\frac{3}{4}}{\sum_{j=k}^{n-1}j|T_j|+ \frac{n}{v}|T_n|(I+\frac{1}{4})}
        ]
   \end{align*}

    Since $\sum_{j=k}^{n-1}j|T_j| \geq (n-1)|T_{n-1}|\geq (n-1)\frac{|T_n|}{2}$,
    we have
    \begin{align*}
    \sum_{j=k}^{n-1}j|T_j| + \frac{n}{v}|T_n|(I+\frac{1}{4})
    &\geq \frac{n-1}{2}|T_n|+\frac{n}{v}|T_n|(I+\frac{1}{4})\\
    &= \frac{n|T_n|}{2v}(v-\frac{v}{n}+2I+\frac{1}{2}).
    \end{align*}
    Therefore
    \begin{align*}
    \limsup_{n\rightarrow\infty} \frac{|\Sigma|^{pn}}{\sum_{j=k}^{n-1}j|T_j|+ \frac{n}{v}|T_n|(I+\frac{1}{4})}
    &\leq \limsup_{n\rightarrow\infty} \frac{|\Sigma|^{pn}}{\frac{(n-1)}{2}|T_n|}\\
    &\leq \limsup_{n\rightarrow\infty} \frac{|\Sigma|^{pn}}{|\Sigma|^{an}} = 0
    \end{align*}
    and
    $$
        \frac{1}{2v'}\frac{\sum_{j=k}^{n-1}j|T_j|}{\sum_{j=k}^{n-1}j|T_j|+ \frac{n}{v}|T_n|(I+\frac{1}{4})}
        \leq \frac{1}{2v'}
    $$
    which is arbitrarily small by choosing $v'$ accordingly, and
    $$
    \frac{n|T_n|}{2v}\frac{\frac{I}{v'}+\frac{3}{4}}{\sum_{j=k}^{n-1}j|T_j|+ \frac{n}{v}|T_n|(I+\frac{1}{4})}
    \leq \frac{\frac{I}{v'}+\frac{3}{4}}{v-\frac{v}{n}+2I+1}
    $$
    which is arbitrarily small by choosing $v$ accordingly.
    Thus
   $$
        R_{PD}(S)
        = \limsup_{n\rightarrow\infty} \frac{|p_n|}{|S[1\ldots n]|}  \leq \frac{1}{2}.
    $$

        For the sake of completeness we give a detailed description of $C$. Let $Q$ be the following set of states:
        \begin{itemize}
            \item   The start state $q_0$, and $q_1,\ldots q_w \ $ the ``early'' states that will count up to
                    $$w=|S_1 S_2 \ldots S_{k-1} \ 1^k \ 1^{k+1} \ \ldots 1^{2k-1}|.$$
            \item   $q^A_0, \ldots, q^A_k \quad$ the $A$ zone states that cruise through the $A$ zone up to the first flag.
            \item   $q^{f}_{j} \quad$  the $j$th flag state, ($j=1,\ldots,v+1$)
            \item   $q^{X_j}_0, \ldots, q^{X_j}_k \quad$ the $X_j$ zone states that cruise through the $X_j$
            zone, pushing every symbol on the stack,
                    until the $(j+1)$-th flag is met, ($j=1,\ldots,v$).
            \item   $q^{Y_j}_1, \ldots, q^{Y_j}_{v'} \quad$ the $Y_j$ zone states that cruise through the $Y_j$
            zone, popping an symbol  from the stack (per input symbol) and comparing it to the input symbol,
             until the stack is empty, ($j=1,\ldots,v$).
                \item   $q^{r,j}_0, \ldots, q^{r,j}_k \quad$ which after the $j$th flag is detected,
            pop $k$ symbols from the stack that were
                    erroneously pushed while reading the $j$th flag, ($j=2,\ldots,v+1$).
            \item $q_{e}, q_{e'}\quad$ the error states, if one symbol of $Y_i$ is not equal to the content of the stack.
        \end{itemize}
        We next describe the transition function $\delta :Q \times \Sigma^{*} \times\Sigma^{*}\rightarrow Q\times\Sigma^{*}$.
        First $\delta$ counts up to $w$ i.e. for $i=0,\ldots w-1$
        $$
            \delta(q_i,x,y) = (q_{i+1},y) \quad \text{ for any } x,y
        $$
        and after reading $w$ symbols, it enters in the first $A$ zone state, i.e. for any $x,y$
        $$\delta(q_w,x,y) = (q^A_0,y).$$
        Then $\delta$ skips through $A$ until the string $1^k$ is met, i.e. for $i=0,\ldots k-1$ and any $x,y$
        $$
            \delta(q^A_i,x,y) =
        \begin{cases}
        (q^A_{i+1},y) &\text{ if } x=1\\
        (q^A_{0},y) &\text{ if } x\ne 1\\
        \end{cases}
        $$ and
        $$
        \delta(q^A_k,x,y) = (q^{f}_1,y).
        $$
        Once $1^k$ has been seen, $\delta$ knows the first flag has started, so it skips
        through the flag until a zero is met, i.e. for every $x,y$
        $$
            \delta(q^{f}_1,x,y) =
        \begin{cases}
        (q^{f}_1,y) &\text{ if } x=1\\
        (q^{X_1}_{0},0y) &\text{ if } x= 0\\
        \end{cases}
        $$
        where state $q^{X_1}_0$ means that the first symbol of the $X_1$ zone (a zero symbol) has been read, therefore $\delta$ pushes a zero.
        In the $X_1$ zone, delta pushes every symbol it sees until it reads a sequence of $k$ ones, i.e up to the start of the second flag, i.e
        for $i=0,\ldots k-1$ and any $x,y$
        $$
            \delta(q^{X_1}_i,x,y) =
        \begin{cases}
        (q^{X_1}_{i+1},xy) &\text{ if } x=1\\
        (q^{X_1}_{0},xy) &\text{ if } x\ne 1\\
        \end{cases}
        $$
        and
        $$
        \delta(q^{X_1}_k,x,y) = (q^{r,2}_0,y).
        $$
        At this point, $\delta$ has pushed all the $X_1$ zone on the stack, followed by
        $k$ ones. The next step is to pop $k$ ones,
        i.e
        for $i=0,\ldots k-1$ and any $x,y$
        $$
            \delta(q^{r,2}_i,x,y) = (q^{r,2}_{i+1},\lambda)
        $$
        and
        $$
        \delta(q^{r,2}_k,x,y) = (q^{f}_2,y).
        $$
        At this stage, $\delta$ is still in the second flag (the second flag is always bigger than $2k$)
        therefore it keeps on reading ones until a zero (the first symbol of the $Y$ zone) is met. For any $x,y$
        $$
            \delta(q^{f}_2,x,y) =
        \begin{cases}
        (q^{f}_2,y) &\text{ if } x=1\\
        (q^{Y_1}_1,\lambda) &\text{ if } x=0 .
        \end{cases}
        $$
        On the last step, $\delta$ has read the first symbol of the $Y_1$ zone, therefore it pops it. At this stage,
        the stack  exactly contains the $X_1$ zone written in reverse order (except the first symbol),
        $\delta$ thus uses its stack to check that what follows is really the $Y_1$ zone. If it is not the case,
            it enters $q_e$. While cruising through $Y_1$, $\delta$  counts with period $v'$. Thus
            for $i=1,\ldots v'-1$ and any $x,y$
        $$
            \delta(q^{Y_1}_i,x,y) =
        \begin{cases}
        (q^{Y_1}_{i+1},\lambda) &\text{ if } x=y\\
        (q_e,\lambda) &\text{ otherwise } \\
        \end{cases}
        $$
        and
            $$
            \delta(q^{Y_1}_{v'},x,y) =
        \begin{cases}
        (q^{Y_1}_{1},\lambda) &\text{ if } x=y\\
        (q_e,\lambda) &\text{ otherwise } \\
        \end{cases}
        $$

         Once the stack is empty,
        the $X_2$ zone begins. Thus, for any $x,y$, $1\leq i \leq v'$
        $$\delta(q^{Y_1}_i,x,z_0) =
        \begin{cases}
        (q^{X_2}_1,1z_0) &\text{ if } x=1\\
        (q^{X_2}_0,0z_0) &\text{ if } x=0 .
        \end{cases}
        $$

    Then for $2\leq j \leq v$ the states corresponding to the $X_j$
    and $Y_j$ zones behave similarly (that is, states $q^{X_j}_{i}$, $q^{r,j+1}_i$,
    $q^{f}_{j+1}$, and  $q^{Y_j}_{i}$).

    At the end of $Y_v$, a new $A$ zone starts, thus for any $1\leq i \leq
    v'$
    $$
            \delta(q_i^{Y_v},x,z_0) =
        \begin{cases}
        (q^A_1,z_0) &\text{ if } x=1\\
        (q^A_0,z_0) &\text{ if } x=0 .
        \end{cases}
        $$

        Once in the $q_e$ state the compressor  outputs a flag then enters state $q_{e'}$, from that point it simply outputs the input, thus
        $$
            \delta(q_e,\lambda,\lambda) = (q_{e'},\lambda)
        $$
            and
            $$\delta(q_{e'},x,y) = (q_{e'},y)$$

        The output  function outputs the input on every state, except on states $q^{Y_j}_1, \ldots, q^{Y_j}_{v'} \quad$            $(j=1,\ldots,v)$
        where for $1\leq i <v'$
        $$
            \nu(q^{Y_j}_{i},b,y) = \lambda
        $$
        and
        $$
            \nu(q^{Y_j}_{v'},b,y) = 0
        $$
          and $q_e$ where a flag is output i.e.,
            $$\nu(q_e,\lambda,\lambda)=10.$$

      Finally, with a similar construction as in the proof of
      Theorem~\ref{theo:ipd}, the inverse of $C$ can be computed by a pushdown
      compressor, showing that $C$ is invPD.  

\end{proof}

\subsection{plogon beats Lempel Ziv}

Our next result uses a Copeland-Erd\"os sequence
\cite{Cham33,CopErd46} on which Lempel-Ziv has maximal
compression ratio, whereas with logspace each prefix of the
sequence can be completely reconstructed from its length.

\begin{theorem}
    There exists a sequence $S$ such that
    $$R_{\pl}(S)=0\quad \text{and }\rho_{LZ}(S)=1.$$
\end{theorem}
\begin{proof}
    Let $S=E(\Sigma^{*})$ be the enumeration of  strings over
    $\Sigma$
    in the standard lexicographical order.
    LZ does not compress $S$ at all, for this algorithm it is the worst possible case, i.e.
    $$\rho_{LZ}(S)=1.$$

    For any input $w$, with $|w|=n$, let $m\in\N$, $x\in\Sigma^{*}$ be such that $w=S[1\ldots m]x$, and $S[1\ldots m+1]\not\sqsubset w$.
Then we define compressor $C$ as
    $C(w,|w|)=\mathrm{dbin}(m) 01 x$,
    where $\mathrm{dbin}(m)$ is $m$ written in binary with every bit doubled
    (such that the separator 01 can be recognized).
    $C$ is clearly 1-1. $C$ is plogon, because on input $(w,n)$, $C$ reads the input online
    to check that $w$ is a prefix of $S$ (i.e. the standard enumeration of strings over
    $\Sigma$);
    the biggest string to check has size $\log n$, therefore the check can be done in plogon.
    As soon as the check fails, $C$ outputs the length (in binary, with every bit doubled)
    of the prefix of the input that satisfied the check
    (at most $2\log n$ bits) followed by 01 and the rest of the input.

    The worst case compression ratio for sequence $S$ is given by
    $$
    R_{\pl}(S)=\limsup_{n\rightarrow\infty}\frac{|C(S[1\ldots n],n)|}{n}
    =\limsup_{n\rightarrow\infty}\frac{2\log n}{n} =0.
    $$
\end{proof}

\subsection{plogon beats Pushdown compressors}\label{section:last}
The next result shows that plogon compressors outperform
our most general family of pushdown compressors on certain
sequences.

The proof is an extension of the intuition in Theorem
\ref{theo:first}, from a few Kolmogorov-random strings a
much longer pushdown-incompressible string can be
constructed, even if an identifying index for each string
is included. The index can then be used by a
polylogarithmic compressor to compress optimally the
sequence.

\begin{theorem}\label{theo:last}
    There exists a sequence $S$ such that
    $$R_{\pl}(S) = 0 \quad \text{and } \rho_{\pd}(S)=1.$$
\end{theorem}

\begin{proof}
    Consider the sequence $S=S_1S_2\ldots$ where $S_n$ is constructed as follows.
    Let $x=x_1 x_2 \ldots x_{n^2}$ ($|x_i|=n$) be a random string with $K(x)\geq n^3\log |\Sigma|$.
    Let
    $$S_n= x_1 x_2 \ldots x_{n^2}i_1x_{i_1}i_2x_{i_2}\ldots i_{2^n}x_{i_{2^n}}$$
    where $i_j \in\{1,\ldots n^2\}$ for every $1\leq j \leq 2^n$ are indexes coded in $2\log n$
    bits, defined later on.

 Let $C_1, C_2,\ldots$ be an enumeration of all  ILPDCwE such that
    $C_i$ can be encoded in at most $i$ bits and such that
            a maximum of
            $\log^{(2)}i$ $\lambda$-rules can be applied per symbol.

    The following claim shows that there are many $C$-incompressible
    strings $x_i$.

    \begin{claim}\label{cd.many.incompressible}
    Let $F_n=\{C_1,\ldots,C_{\log n}\}$.
    Let $w\in\Sigma^*$.
    \begin{enumerate}
\item Let
    $C\in F_n$.    There are at least $(1-\frac{1}{2\log n})n^2$ strings $ix_i$
($1\leq i \leq n^2$) such that
    $$|C(wix_i)|-|C(w)| > n-2\sqrt{n} .$$

    \item There is a  string $x_i$ such that  for every  $C\in F_n$, $$|C(wix_i)|-|C(w)| > n-2\sqrt{n} .$$
\end{enumerate}
    \end{claim}
\begin{proofof}{of Claim \ref{cd.many.incompressible}}
After having read $w$,
    $C$ is in state $q$, with stack content $yz$, where $y$ denotes the $(n+2\log n)\log^{(2)}n$
    topmost symbols of the stack (if the stack is shorter then $y$ is the whole stack).
    It is clear that while reading an $ix_{i}$, $C$ will not pop the stack below $y$.

    Let $T=(1-\frac{1}{2\log n})n^2$, and let $C(q,yz,ix_i\$)$
    denote the output of $C$ when started in state $q$ on input $ix_i\$$ with stack content $yz$.
    Suppose the claim false, i.e.
    there exist more than $n^2-T$ words $ix_i$ such that $C(q,yz,ix_i\$)=p_i$, ends in state $q_i$,
    and $|p_i|\leq n-2\sqrt n+O(1)$ (notice that the output on symbol $\$$ is $O(1)$). Denote by $G$ the set of such strings $x_i$.
    This yields the following short program for $x$ (coded with alphabet $\Sigma$):
    $$p=(n,C,q,y,a_1t_1a_2t_2\ldots a_{n^2}t_{n^2})$$
    where each comma costs less than $3\log |s|$,   where $s$ is the element between two commas;
    $a_i=1$ implies $t_i=x_i$, $a_i=0$ implies $x_i\in G$ and
    $t_i = d(q_i)01d(|p_i|)01p_i$ (where $d(z)$ for any string $z$, is the string written with every symbol doubled),
    i.e. $|t_i|\leq n - \sqrt n$.
    $p$ is a program for $x$: once $n$ is known, each $a_i t_i$ yields either $x_i$ (if $a_i=1$)
    or $(p_i,q_i)$ (if $a_i=0$). From $(p_i,q_i)$, simulating $C(q,yz,u\$)$ for each $u\in\Sigma^{n+2 \log n}$
    yields the unique $u=ix_i$ such that $C(q,yz,u\$)=p_i$ and ends in state $q_i$. The simulations are possible,
    because $C$ does not read its stack further than $y$, which is given.
    We have
    \begin{align*}
    |p|&\leq O(\log n) + (n+2\log n)\log^{(2)}n+(n+1)T+(n^2-T)(n-\sqrt n)\\
    &\leq O(n^2) + n^3 - \frac{n^{2.5}}{2\log n}\\
    &\leq n^3 -\frac{n^{2.5}}{4\log n} \\
    \end{align*}
    which contradicts the randomness of $x$, thus proving part 1. 

    Let $W_j$ be the set of strings $ix_i$ that are compressible by $C_j$; by
        1.,
    $|W_j|\leq n^2/2\log n$. Let $R=\{ix_{i}\}_{i=1}^{n^2} - \cup_{j=1}^{\log n}W_j$
    be the set of strings incompressible by all $C\in F_n$. We have
    $$|R|\geq n^2 - \log n \cdot n^2 / 2\log n =n^2/2> 1.$$
    This proves part 2.
\end{proofof}

       We finish the definition of $S_n$ by picking
     $i_1x_{i_1}$ to be the first string fulfilling the second part of Claim \ref{cd.many.incompressible} for $w=S_1S_2\ldots S_{n-1}$.
    The construction is similar for all strings $\{x_{i_j}\}_{j=2}^{2^n}$, by taking $w=S_1S_2\ldots S_{n-1}x_{i_1}\ldots x_{i_{j-1}}$, thus ending the construction
    of $S_n$.

    Let us show that $\rho_{\pd}(S)=1$. Let $\epsilon>0$.
    Let $C=C_k$ be an ILPDCwE; then for almost every $n$, and for all $0\leq t \leq 2^n$,
    because $|S_1\ldots S_{n-1}|$ is exponentially larger than
    the first $n^2$ $x_i$'s of zone $S_n$, it is good enough to compute the compression ratio
    only after those first $n^2$ $x_i$'s and after each $ix_i$. We have
    \begin{align*}
    &\frac{|C(S_1\ldots S_{n-1}S_n[n^2+t(n+2\log n)]\$)|}{|S_1\ldots S_{n-1}S_n[n^2+t(n+2\log n)]|} \\
    &\geq
    \frac{\sum_{j=k}^{n-1}(2^j)(j-2\sqrt j)+ t(n-2\sqrt n)}
    {\sum_{j=1}^{n-1}(j^2+2^j(j+2\log j))+n^2 +t(n+2\log n)}\\
    &\geq  \frac{(1-\alpha)\sum_{j=1}^{n-1}j2^j+n^2+tn}
    {(1+\alpha)\sum_{j=1}^{n-1}j2^j+n^2+tn} - \frac{2(t+1)\sqrt n}{\sum_{j=1}^{n-1}j2^j+n^2+tn}
    -\frac{(1-\alpha)\sum_{j=1}^{k}j2^j}{\sum_{j=1}^{n-1}j2^j+n^2+tn}\\
    &\geq 1 -\epsilon/4 -\epsilon/4 -\epsilon/4 \geq 1 - \epsilon
    \end{align*}
    where $\alpha$ can be made arbitrarily small for large enough $n$.

    We show that $R_{\pl}(S)=0$.    Consider the following $\pl$ compressor $C$, where every output bit  is output
    doubled except commas (coded by $10$) and the error flag (coded by $01$).
    First $C$ outputs the length of the input (in binary) followed by a comma.
    For the $n^2$ first $x_i$'s of zone $S_n$, $C$ outputs them (and stores them).
    For the remaining $i_jx_{i_j}$'s, only $i_j$ is output, and $C$ checks that what follows
    $i_j$ is indeed $x_{i_j}$. If at any point in time the test fails, the error mode is entered.
    In error mode, $01$ is output, followed by the rest of the input, starting right after the  $i_j$
    where the error occurred.

    It is easy to check that $C$ is polylog space, since at the beginning of zone $S_n$, the available space
    is of order $\text{poly}(n)$.

    $C$ is IL, because from $C$'s output, we know the length of the input and whether the error mode has been
    entered or not. If there is no error, all the first $n^2$  $x_i$'s of zone $S_n$ can be recovered,
    followed by all strings $i_jx_{i_j}$. If the error mode is entered, by the previous argument the sequence $S_n$ can
    be reconstructed up to the last $i_j$ before the error. The rest of the output yields the rest of the sequence.

    Let us compute the compression ratio. Let $\epsilon >0$. Let $n\in\N$ and $0\leq t \leq 2^n$.
    Because $|S_1\ldots S_{n-1}|$ is exponentially larger than
    the first $n^2$ $x_i$'s of zone $S_n$, it is good enough to compute the compression ratio
    only after those first $n^2$ $x_i$'s.
    We have
    \begin{align*}
    \frac{|C(S_1\ldots S_{n-1}S_n[n^3+t(n+2\log n)])|}{|S_1\ldots S_{n-1}S_n[n^3+t(n+2\log n)]|} &\leq
    \frac{2[\sum_{j=1}^{n-1}j^3+2^j(2\log j)+n^3+2t\log n]}
    {\sum_{j=1}^{n-1}j^3+2^j(j+2\log j)+n^3+t(n+2\log n)}\\
    &\leq \frac{2[\sum_{j=1}^{n-1}3\cdot2^j\log j+n^3+2t\log n]}
    {\sum_{j=1}^{n-1}j2^j+n^3+tn}\\
    &\leq \frac{6[\sum_{j=1}^{n-1}2^j\log j]}{\sum_{j=1}^{n-1}j2^j} + \epsilon/4 + \epsilon/4\\
    \end{align*}
    Since $\log j < \frac{\epsilon}{24} j$ for all $j>j_0$ we have
    \begin{align*}
    \frac{|C(S_1\ldots S_{n-1}S_n[n^3+t(n+2\log n)])|}{|S_1\ldots S_{n-1}S_n[n^3+t(n+2\log n)]|}
    &\leq \frac{6[\sum_{j=1}^{j_0}2^j\log j]}{\sum_{j=1}^{n-1}j2^j} +
    \frac{\epsilon/4[\sum_{j=j_0+1}^{n-1}j2^j ]}{\sum_{j=1}^{n-1}j2^j} + \epsilon/2 \\
    &\leq \epsilon/4 + \epsilon/4 + \epsilon/2 \leq \epsilon .
    \end{align*}
\end{proof}

   \section{Conclusion}\label{sec:conclusions}
   The equivalence of compression ratio, effective dimension, and
log-loss unpredictability has been explored in different
settings \cite{FSD,Hitchcock:PHD,EFDAIT}. It is known that
for the cases of finite-state, polynomial-space, recursive,
and constructive resource-bounds, natural definitions of
compression and dimension coincide, both in the case of
infinitely often compression, related to effective versions
of Hausdorff dimension, and that of almost everywhere
compression, matched with packing dimension. The general
matter of transformation of compressors in predictors and
vice versa is  widely studied \cite{ScuBro06}.

In this paper we have done a complete comparison of
pushdown, plogon compression and LZ-compression. It is
straightforward to construct a prediction algorithm based
on Lempel-Ziv compressor that uses similar computing
resources, and it has been proved in \cite{BPDD} that
bounded-pushdown compression and dimension coincide. This
leaves us with the natural open question of whether each
plogon compressor can be transformed into a plogon
prediction algorithm, for which the log-loss
unpredictability coincides with the compression ratio of
the initial compressor, that is, whether the natural
concept of plogon dimension coincides with plogon
compressibility. A positive answer would get plogon
computation closer to pushdown devices, and a negative one
would make it closer to polynomial-time algorithms, for
which the answer is likely to be negative \cite{DIC}.

\bibliographystyle{plain}

\begin{thebibliography}{10}


\bibitem{BPDD}
P.~Albert, E.~Mayordomo, and P.~Moser.
\newblock Bounded pushdown dimension vs lempel ziv information density.
\newblock Technical Report TR07-051, ECCC: Electronic Coloquium on
  Computational Complexity, 2007.

\bibitem{AMMP08}
P.~Albert, E.~Mayordomo, P.~Moser, and S.~Perifel.
\newblock Pushdown compression.
\newblock In {\em Proceedings of the 25th Symposium on Theoretical Aspects of
  Computer Science (STACS 2008)}, pages 39--48, 2008.

\bibitem{AlMaSz99}
N.~Alon, Y.~Matias, and M.~Szegedy.
\newblock The space complexity of approximating the frequency moments.
\newblock {\em Journal of Computer and System Sciences}, 58:137--147, 1999.

\bibitem{AluMad06}
R.~Alur and P.~Madhusudan.
\newblock Adding nesting structure to words.
\newblock In {\em Proceedings of the Tenth International Conference on
  Developments in Language Theory}, volume 4036 of {\em Lecture Notes in
  Computer Science}. Springer, 2006.

\bibitem{AuBeBo07}
J.~Autebert, J.~Berstel, and L.~Boasson.
\newblock Context-free languages and pushdown automata.
\newblock In G.~Rozenberg and A.~Salomaa, editors, {\em Handbook of Formal
  Languages, volume 1, Word, Language, Grammar}, pages 111--174.
  Springer-Verlag, 1997.

\bibitem{Cham33}
D.~G. Champernowne.
\newblock Construction of decimals normal in the scale of ten.
\newblock {\em J. London Math. Soc.}, 2(8):254--260, 1933.

\bibitem{CopErd46}
A.H. Copeland and P.~Erd{\"o}s.
\newblock Note on normal numbers.
\newblock {\em Bulletin of the American Mathematical Society}, 52:857--860,
  1946.

\bibitem{FSD}
J.~J. Dai, J.~I. Lathrop, J.~H. Lutz, and E.~Mayordomo.
\newblock Finite-state dimension.
\newblock {\em Theoretical Computer Science}, 310:1--33, 2004.

\bibitem{GinRos66}
S.~Ginsburg and G.~F. Rose.
\newblock Preservation of languages by transducers.
\newblock {\em Information and Control}, 9(2):153--176, 1966.

\bibitem{GinRos68}
S.~Ginsburg and G.~F. Rose.
\newblock A note on preservation of languages by transducers.
\newblock {\em Information and Control}, 12(5/6):549--552, 1968.

\bibitem{HarSha06}
S.~Hariharan and P.~Shankar.
\newblock Evaluating the role of context in syntax directed compression of xml
  documents.
\newblock In {\em Proceedings of the 2006 IEEE Data Compression Conference (DCC
  2006)}, page 453, 2006.

\bibitem{HaImMa78}
J.~Hartmanis, N.~Immerman, and S.~Mahaney.
\newblock One-way log-tape reductions.
\newblock In {\em Proceedings of the 19th Annual Symposium on Foundations of
  Computer Science (FOCS'78)}, pages 65--72. IEEE Computer Society, 1978.

\bibitem{Hitchcock:PHD}
J.~M. Hitchcock.
\newblock {\em Effective Fractal Dimension: Foundations and Applications}.
\newblock PhD thesis, Iowa State University, 2003.

\bibitem{IndWoo05}
P.~Indyk and D.P. Woodruff.
\newblock Optimal approximations of the frequency moments of data streams.
\newblock In {\em Proceedings of the 37th Annual ACM Symposium on Theory of
  Computing (STOC 2005)}, pages 202--208. ACM, 2005.

\bibitem{KuMaVi07}
V.~Kuma, P.~Madhusudan, and M.~Viswanathan.
\newblock Visibly pushdown automata for streaming xml.
\newblock In {\em International World Wide Web Conference WWW 2007}, pages
  1053--1062, 2007.

\bibitem{LatStr97}
J.~I. Lathrop and M.~J. Strauss.
\newblock A universal upper bound on the performance of the {L}empel-{Z}iv
  algorithm on maliciously-constructed data.
\newblock In B.~Carpentieri, editor, {\em Compression and Complexity of
  Sequences '97}, pages 123--135. IEEE Computer Society Press, 1998.

\bibitem{LeaEng07}
C.~League and K.~Eng.
\newblock Type-based compression of xml data.
\newblock In {\em Proceedings of the 2007 IEEE Data Compression Conference (DCC
  2007)}, pages 272--282, 2007.

\bibitem{LemZiv78}
A.~Lempel and J.~Ziv.
\newblock Compression of individual sequences via variable rate coding.
\newblock {\em IEEE Transaction on Information Theory}, 24:530--536, 1978.

\bibitem{DIC}
M.~L{\'o}pez-Vald{\'e}s and E.~Mayordomo.
\newblock Dimension is compression.
\newblock In {\em Proceedings of the 30th International Symposium on
  Mathematical Foundations of Computer Science}, volume 3618 of {\em Lecture
  Notes in Computer Science}, pages 676--685. Springer-Verlag, 2005.

\bibitem{EFDAIT}
E.~Mayordomo.
\newblock Effective fractal dimension in algorithmic information theory.
\newblock In {\em New Computational Paradigms: Changing Conceptions of What is
  Computable}, pages 259--285. Springer-Verlag, 2008.

\bibitem{MayMos09}
E.~Mayordomo and P.~Moser.
\newblock polylog space compression is incomparable with lempel-ziv and
  pushdown compression.
\newblock In {\em Proceedings of the 35th International Conference on Current
  Trends in Theory and Practice of Computer Science (SOFSEM09)}, volume 5404,
  pages 633--644. Springer Lecture Notes in Computer Science, 2009.

\bibitem{ScuBro06}
D.~Sculley and C.~E. Brodley.
\newblock Compression and machine learning: A new perspective on feature space
  vectors.
\newblock In {\em Proceedings of the Data Compression Conference (DCC-2006)},
  pages 332--341, 2006.

\end{thebibliography}


\end{document}